\documentclass[runningheads]{llncs}
\usepackage{graphicx}
\usepackage{amsmath}
\usepackage[capitalise]{cleveref}%
\usepackage{float}
\usepackage{amsfonts}
\usepackage{tikz}
\usepackage{mathrsfs}
\usepackage{amssymb}
\usepackage{adjustbox}
\usepackage{enumitem}
\usepackage{flafter}
\usepackage{mathtools}
\usetikzlibrary{arrows,decorations.pathreplacing,backgrounds,calc,positioning}
\allowdisplaybreaks
\crefname{case}{Case}{Cases}
\usepackage[ruled,vlined,linesnumbered]{algorithm2e}
\usepackage{comment}
\newcommand{\convexpath}[2]{
	[   
	create hullnodes/.code={
		\global\edef\namelist{#1}
		\foreach [count=\counter] \nodename in \namelist {
			\global\edef\numberofnodes{\counter}
			\node at (\nodename) [draw=none,name=hullnode\counter] {};
		}
		\node at (hullnode\numberofnodes) [name=hullnode0,draw=none] {};
		\pgfmathtruncatemacro\lastnumber{\numberofnodes+1}
		\node at (hullnode1) [name=hullnode\lastnumber,draw=none] {};
	},
	create hullnodes
	]
	($(hullnode1)!#2!-90:(hullnode0)$)
	\foreach [
	evaluate=\currentnode as \previousnode using \currentnode-1,
	evaluate=\currentnode as \nextnode using \currentnode+1
	] \currentnode in {1,...,\numberofnodes} {
		let
		\p1 = ($(hullnode\currentnode)!#2!-90:(hullnode\previousnode)$),
		\p2 = ($(hullnode\currentnode)!#2!90:(hullnode\nextnode)$),
		\p3 = ($(\p1) - (hullnode\currentnode)$),
		\n1 = {atan2(\y3,\x3)},
		\p4 = ($(\p2) - (hullnode\currentnode)$),
		\n2 = {atan2(\y4,\x4)},
		\n{delta} = {-Mod(\n1-\n2,360)}
		in 
		{-- (\p1) arc[start angle=\n1, delta angle=\n{delta}, radius=#2] -- (\p2)}
	}
	-- cycle
}

\begin{document}
\title{Characterization of the Imbalance Problem on Complete Bipartite Graphs
}
%
%
\author{Steven Ge
	\and
	Toshiya Itoh
}
%
%
\institute{Tokyo Institute of Technology, Meguro, Japan\\
\email{ge.s.aa@m.titech.ac.jp}\\
\email{titoh@c.titech.ac.jp}}
\maketitle              
\begin{abstract}
	We study the imbalance problem on complete bipartite graphs. The imbalance problem is a graph layout problem and is known to be NP-complete.
	Graph layout problems find their applications in the optimization of  networks for parallel computer architectures, VLSI circuit design,  information retrieval, numerical analysis, computational biology, graph theory, scheduling and archaeology \cite{diaz2002survey}.
	In this paper, we give characterizations for the optimal solutions of the imbalance problem on complete bipartite graphs. Using the characterizations, we can solve the imbalance problem in time polylogarithmic in the number of vertices, when given the cardinalities of the parts of the graph, and verify whether a given solution is optimal in time linear in the number of vertices on complete bipartite graphs. We also introduce a generalized form of complete bipartite graphs on which the imbalance problem is solvable in time quasilinear in the number of vertices by using the aforementioned characterizations.
	
	
	\keywords{Imbalance Problem  \and Vertex layout \and Complete bipartite graph \and Proper interval bipartite graph.}
\end{abstract}
\section{Introduction}  
Graph layout problems are combinatorial optimization problems, where the goal is to find an ordering on the vertices that optimizes an objective function. A large number of problems from different domains can be formulated as graph layout problems \cite{diaz2002survey}. The imbalance problem is a graph layout problem that has applications in 3-dimensional circuit design \cite{wood}.

The imbalance problem was introduced by Biedl et al. \cite{BIEDL200527}. Given an ordering of the vertices of a graph $G$, the imbalance of a vertex $v$ is the absolute difference in the number of neighbors to the left of $v$ and the number of neighbors to the right of $v$. The imbalance of an ordering is the sum of the imbalances of the vertices. An instance of the imbalance problem consists of a graph $G$ and an integer $k$. The problem asks whether there exists an ordering on the vertices of $G$ such that the imbalance of the ordering is at most $k$.

The imbalance problem is NP-complete for several graph classes, including bipartite graphs with degree at most 6, weighted trees \cite{BIEDL200527}, general graphs with degree at most 4 \cite{10.1007/11533719_86}, and split graphs \cite{10.1007/978-3-030-26176-4_18}. 
The problem becomes polynomial time solvable on superfragile graphs \cite{10.1007/978-3-030-26176-4_18}. The problem is linear time solvable on proper interval graphs \cite{10.1007/978-3-030-26176-4_18}, bipartite permutation graphs, and threshold graphs \cite{10.1007/978-3-030-64843-5_52}.



Gorzny showed that the minimum imbalance of a bipartite permutation graph $G = (V,E)$, which class is a superclass of complete bipartite graphs and proper interval bipartite graphs, can be computed in $\mathcal{O}(|V| + |E|)$ time \cite{10.1007/978-3-030-64843-5_52}. 
We give characterizations for the optimal solutions of the imbalance problem on complete bipartite graphs. Using the characterizations, we show that the imbalance problem is solvable in $\mathcal{O}(\log(|V|) \cdot \log(\log(|V|)))$ time on complete bipartite graphs, when given the cardinalities of the parts of the graph. Additionally, using the characterizations, we can verify whether a given solution is optimal in $O(|V|)$ on complete bipartite graphs. 
We also introduce a generalized form of complete bipartite graphs, which we call chained complete bipartite graphs, on which the imbalance problem is solvable in $\mathcal{O}(c \cdot \log(|V|) \cdot \log(\log(|V|)))$ time, where $c = \mathcal{O}(|V|)$, by using the aforementioned characterizations. As chained complete bipartite graphs are a subclass of proper interval bipartite graphs, the result of Gorzny also applies to chained complete bipartite graphs.
\section{Preliminaries}\label{c2}


We only consider graphs that are finite, undirected, connected, and simple (i.e. without multiple edges or loops). A graph $G$ is denoted as $G = (V, E)$, where $V$ denotes the set of vertices and $E \subseteq V \times V$ denotes the set of undirected edges. For convenience, we abbreviate bipartite graph as bigraph.




For $n \in \mathbb{N}^+$, let us define 
$[n] = \{1, 2, \dots, n\}.$


We define $\sigma_{S}: S \rightarrow [|S|]$ to be an ordering of $S$. For convenience, at times we denote an ordering $\sigma_S$ by 
$(\sigma^{-1}(1), \sigma^{-1}(2), \dots, \sigma^{-1}(|S|)),$ 
where, $v = \sigma^{-1}(\sigma(v))$  for $v \in S$.
We also say that $v \in S$ is at position $k$ in ordering $\sigma_S$ if $\sigma_S(v) = k$.

Let $S_1, S_2, \dots, S_n$ be a collection of disjoint sets. Let $s \in S_i$, where $1\leq i \leq n$, then the concatenation of orderings is defined as follows: 
$\sigma_{S_1}\sigma_{S_2}\dots\sigma_{S_n}(s) = \sigma_{S_i}(s) + \sum\limits_{j=1}^{i-1}|S_j|.$
Additionally, we use the product notation to denote the concatenation over a set. That is,
$\prod\limits_{i=1}^{n}\sigma_{S_i} = \sigma_{S_1}\sigma_{S_2}\dots\sigma_{S_n}.$

Given an ordering $\sigma_S$ we say that $v \in S$ occurs to the left of $u \in S$ in $\sigma_S$, if and only if $\sigma_S(v) < \sigma_S(u)$. We denote this as
$v <_{\sigma_S} u.$
We define $>_{\sigma_S}$ analogously.

Given an ordering $\sigma_S$ we say that $\sigma_S^{S'}$, where $S' \subseteq S$, is a subordering of $\sigma_S$ on $S'$ if $\sigma_S^{S'}$ preserves the relative ordering of the elements in $S'$. That is $\sigma_S^{S'}$ is a subordering of $\sigma_S$ if and only if
$\forall_{u,v \in S'}\forall_{\oplus \in \{<,>\}}~u \oplus_{\sigma_S^{S'}} v \iff u \oplus_{\sigma_S} v.$

For a graph $G=(V,E)$, the open neighborhood of a vertex $v \in V$, denoted by $N(v)$, is the set of vertices adjacent to $v$. We call the vertices in $N(v)$ the neighbors of vertex $v$. That is
$N(v) = \{u \in V ~|~ \{u,v\} \in E\}.$

The imbalance of a vertex $v \in V$ on the ordering $\sigma_{V}$ in graph $G=(V,E)$, denoted by $I(v, \sigma_{V}, G)$, is defined to be the absolute difference in the number of neighbors of $v$ occurring to the left of $v$ and the number of neighbors of $v$ occurring to the right of $v$ in $\sigma_V$. That is,
$I(v, \sigma_{V}, G) = \big||\{u \in N(v) ~|~ u <_{\sigma_V} v\}| - |\{u \in N(v) ~|~ u >_{\sigma_V} v\}|\big|.$

The imbalance of an ordering $\sigma_{V}$ on graph $G=(V,E)$, denoted by $I(\sigma_{V}, G)$, is defined to be sum over the imbalances of the vertices in $V$ on the ordering $\sigma_{V}$ in graph $G$. That is
$I(\sigma_{V}, G) = \sum\limits_{v \in V}I(v, \sigma_{V}, G).$
If the ordering and/or graph are clear from the context, we shall exclude them from the parameters of the function $I$.

The imbalance of a graph $G=(V,E)$, denoted by $I(G)$, is defined to be the minimum imbalance over all orderings $\sigma_{V}$. That is,
$I(G) = \min\limits_{\sigma_V \in S(V)}I(\sigma_{V}, G),$
where $S(V)$ denotes the set of all orderings on $V$. 
We call an ordering $\sigma_V$ whose imbalance is equivalent to $I(G)$ an (imbalance) optimal ordering.
Given a graph $G$ and an integer $k$, the imbalance problem asks whether $I(G) \leq k$ is true or false.
A bigraph $G = (X,Y,E)$ is an interval bigraph, if there exists a set of intervals on the real line, where for each vertex $v \in X \cup Y$ we have exactly one corresponding interval such that the intervals corresponding to vertices $x \in X$ and $y \in Y$ intersect if and only if there exists an edge in $E$ connecting $x$ and $y$. We call such a set of intervals the interval representation $I_G$ of graph $G$. That is, $I_G = \{[l_v, r_v] ~|~ v \in X \cup Y\}$ such that
$\forall_{x \in X}~\forall_{y \in Y}~ \big([l_x, r_x] \cap [l_y, r_y] \neq \emptyset \iff \{x,y\} \in E\big).$


An interval bigraph $G = (X,Y,E)$ is a proper interval bigraph, abbreviated by PI-bigraph, if it has an interval representation $I_G$ such that none of the intervals are properly contained in another. That is, there exists an interval representation $I_G = \{[l_v, r_v] ~|~ v \in X \cup Y\}$ of $G$ such that	
$\forall_{u,v \in X \cup Y}~(u \neq v \iff [l_u,r_u] \nsubseteq [l_v,r_v]).$
\section{Imbalance on Complete bipartite graphs}\label{c3}

Let $G = (X,Y,E)$ be a complete bigraph. We shall prove that the minimum imbalance of $G$ is $|X|\cdot |Y| + (|X| \mod 2) \cdot (|Y| \mod 2)$. 

\begin{lemma}\label{lmacbg1}
	If $G = (X,Y,E)$ is a complete bigraph, then there exists an ordering $\sigma_{X \cup Y}$ such that
	$I(\sigma_{X \cup Y}) = |X|\cdot |Y| + (|X| \mod 2) \cdot (|Y| \mod 2).$
\end{lemma}
\begin{proof}
	The proof constructs an ordering on $X\cup Y$ and considers two cases, either the cardinality one part is even or none of the cardinalities of the parts are even. 
	\begin{case}
		$(|X| \mod 2 = 0) \vee (|Y| \mod 2 = 0).$\\
		W.l.o.g. assume that $|Y|$ is even. Let $Y_1$ and $Y_2$ partition $Y$ into two sets of equal size. Consider ordering $\sigma_{X \cup Y} = \sigma_{Y_1}\sigma_{X}\sigma_{Y_2}$. In this ordering, the imbalance of any vertex in $X$ is zero and the imbalance of any vertex in $Y$ is $|X|$. Thus the imbalance of ordering $\sigma_{X \cup Y}$ is 
		$I(\sigma_{X \cup Y}) = |X| \cdot |Y|.$
		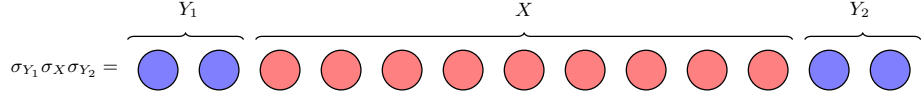
\begin{figure}[H]
			\centering
			\begin{adjustbox}{width=\textwidth}
\begin{tikzpicture}[-,semithick]

\tikzset{Y/.append style={fill=blue!50,draw=black,text=black,shape=circle,minimum size=2em}}
\tikzset{X/.append style={fill=red!50,draw=black,text=black,shape=circle,minimum size=2em}}
\tikzset{t/.append style={fill=white,draw=white,text=black}}
\node[t]         (T) {$\sigma_{Y_1}\sigma_{X}\sigma_{Y_2} =~~~~~~~~~$};
\node[Y]         (M) [right of=T] {$~$};
\node[Y]         (N) [right of=M] {$~$};
\node[X]         (A) [right of=N] {$~$};
\node[X]         (B) [right of=A] {$~$};
\node[X]         (C) [right of=B] {$~$};
\node[X]         (D) [right of=C] {$~$};
\node[X]         (E) [right of=D] {$~$};
\node[X]         (F) [right of=E] {$~$};
\node[X]         (G) [right of=F] {$~$};
\node[X]         (H) [right of=G] {$~$};
\node[X]         (I) [right of=H] {$~$};
\node[Y]         (O) [right of=I] {$~$};
\node[Y]         (P) [right of=O] {$~$};

\draw[decorate,decoration={brace,amplitude=3pt}] 
(0.5,0.5) coordinate (t_k_unten) -- (2.4,0.5) coordinate (t_k_opt_unten); 
\node[t] at (1.5,1) {$Y_1$};
\draw[decorate,decoration={brace,amplitude=3pt}] 
(2.6,0.5) coordinate (t_k_unten) -- (11.4,0.5) coordinate (t_k_opt_unten); 
\node[t] at (7,1) {$X$};
\draw[decorate,decoration={brace,amplitude=3pt}] 
(11.6,0.5) coordinate (t_k_unten) -- (13.5,0.5) coordinate (t_k_opt_unten); 
\node[t] at (12.5,1) {$Y_2$};
\end{tikzpicture}
\end{adjustbox}
			\caption{Example of ``sandwiched'' ordering for $G = K_{4,9}$.}
			\label{ap-sandwichedexample}
		\end{figure}
	\end{case}
	\begin{case}\label{compbigraphoddodd}
		$(|X| \mod 2 = 1) \wedge (|Y| \mod 2 = 1).$\\
		Let $y_m \in Y$ be an element of $Y$. Let $X_1$ and $X_2$ partition $X$ into two sets such that $||X_1| - |X_2|| = 1$ and let $Y_1$ and $Y_2$ partition $Y \setminus \{y_m\}$ into two sets such that $|Y_1| = |Y_2|$. Consider the ordering $\sigma_{X \cup Y} = \sigma_{Y_1}\sigma_{X_1}\sigma_{\{y_m\}}\sigma_{X_2}\sigma_{Y_2}$. The imbalance of any vertex in $X \cup \{y_m\}$ is 1 and the imbalance of any vertex in $Y \setminus \{y_m\}$ is $|X|$. Thus the imbalance of the ordering $\sigma_{X \cup Y}$ is $I(\sigma_{X \cup Y}) = |X| \cdot |Y| + 1$.
		\begin{figure}[H]
			\centering
			\begin{adjustbox}{width=\textwidth}
\begin{tikzpicture}[-,semithick]
  
  \tikzset{Y/.append style={fill=blue!50,draw=black,text=black,shape=circle,minimum size=2em}}
  \tikzset{X/.append style={fill=red!50,draw=black,text=black,shape=circle,minimum size=2em}}
  \tikzset{t/.append style={fill=white,draw=white,text=black}}
  \node[t]         (preT) [] {};
  \node[t]         (T) [right of=preT] {$\sigma_{Y_1}\sigma_{X_1}\sigma_{\{y_m\}}\sigma_{X_2}\sigma_{Y_2} =~~~~~~~~~~~~~~~~~~~~~~$};
  \node[Y]         (A) [right of=T] {$~$};
  \node[Y]         (B) [right of=A] {$~$};
  \node[Y]         (C) [right of=B] {$~$};
  \node[Y]         (D) [right of=C] {$~$};
  \node[X]         (M) [right of=D] {$~$};
  \node[Y]         (E) [right of=M] {$~$};
  \node[X]         (N) [right of=E] {$~$};
  \node[X]         (F) [right of=N] {$~$};
  \node[Y]         (O) [right of=F] {$~$};
  \node[Y]         (G) [right of=O] {$~$};
  \node[Y]         (H) [right of=G] {$~$};
  \node[Y]         (I) [right of=H] {$~$};
  
  \draw[decorate,decoration={brace,amplitude=3pt}] 
    (1.5,0.5) coordinate (t_k_unten) -- (5.4,0.5) coordinate (t_k_opt_unten); 
    \node[t] at (3.5,1) {$Y_1$};
  \draw[decorate,decoration={brace,amplitude=3pt}] 
    (5.6,0.5) coordinate (t_k_unten) -- (6.4,0.5) coordinate (t_k_opt_unten); 
    \node[t] at (6,1) {$X_1$};
    \draw[decorate,decoration={brace,amplitude=3pt}] 
    (6.6,0.5) coordinate (t_k_unten) -- (7.4,0.5) coordinate (t_k_opt_unten); 
    \node[t] at (7,1) {$\{y_m\}$};
    \draw[decorate,decoration={brace,amplitude=3pt}] 
    (7.6,0.5) coordinate (t_k_unten) -- (9.4,0.5) coordinate (t_k_opt_unten); 
    \node[t] at (8.5,1) {$X_2$};
  \draw[decorate,decoration={brace,amplitude=3pt}] 
    (9.6,0.5) coordinate (t_k_unten) -- (13.5,0.5) coordinate (t_k_opt_unten); 
    \node[t] at (11.5,1) {$Y_2$};
\end{tikzpicture}
\end{adjustbox}
			\caption{Example of ``pseudo-sandwich'' ordering for $G = K_{3,9}$.
			}
			\label{ap-interwovenexample}
		\end{figure}
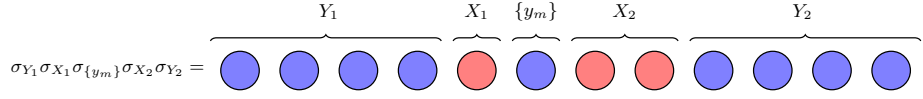
	\end{case}
\end{proof}

\begin{definition}\label{defsL}
	Let $G = (X,Y,E)$ be a complete bigraph and $\sigma_{X\cup Y}$ be an arbitrary ordering on $X\cup Y$. We define $L(Y, \sigma_{X\cup Y}) \subseteq [|X| + |Y|]$ to be the positions of the elements of $Y$ in $\sigma_{X\cup Y}$. That is 
	$L(Y, \sigma_{X\cup Y}) = \{\sigma_{X\cup Y}(y) ~|~ y \in Y\}.$
	Let us denote the elements in $L(Y, \sigma_{X\cup Y})$ as 
	$L(Y, \sigma_{X\cup Y}) = \{l^{\sigma_{X\cup Y}}_1, l^{\sigma_{X\cup Y}}_2, \dots, l^{\sigma_{X\cup Y}}_{|Y|}\}$
	such that $l^{\sigma_{X\cup Y}}_1 < l^{\sigma_{X\cup Y}}_2 < \dots < l^{\sigma_{X\cup Y}}_{|Y|}$. Additionally, we define $l^{\sigma_{X\cup Y}}_0 = 0$ and $l^{\sigma_{X\cup Y}}_{|Y|+1} = |X|+|Y|+1$. We leave out the superscript in $l^{\sigma_{X\cup Y}}_i$, when the ordering is clear from the context. 
	
	Additionally, we define $L^{\sigma_{X\cup Y}}_i$ to be the vertices of $X$ between the positions $l_i$ and $l_{i+1}$ in ordering $\sigma_{X\cup Y}$. Let $Y = \{y_1, \dots, y_{|Y|}\}$ be an arbitrary enumeration of the element in $Y$. 
	We leave out the superscript in $L^{\sigma_{X\cup Y}}_i$, when the ordering is clear from the context.
\end{definition}
\begin{figure}[H]
	\centering
	\begin{adjustbox}{width=\textwidth}
\begin{tikzpicture}[-,semithick]
  
  \tikzset{Y/.append style={fill=blue!50,draw=black,text=black,shape=circle}}
  \tikzset{X/.append style={fill=red!50,draw=black,text=black,shape=circle}}
  \tikzset{t/.append style={fill=white,draw=white,text=black}}
  \node[t]         (T) {$\sigma_{X \cup Y} =~~~$};
  \node[X]         (A) [right of=T] {$x_1$};
  \node[X]         (B) [right of=A] {$x_2$};
  \node[Y]         (M) [right of=B] {$y_1$};
  \node[X]         (C) [right of=M] {$x_3$};
  \node[Y]         (N) [right of=C] {$y_2$};
  \node[X]         (D) [right of=N] {$x_4$};
  \node[X]         (E) [right of=D] {$x_5$};
  \node[X]         (F) [right of=E] {$x_6$};
  \node[Y]         (O) [right of=F] {$y_3$};
  \node[X]         (G) [right of=O] {$x_7$};
  \node[X]         (H) [right of=G] {$x_8$};
  \node[Y]         (P) [right of=H] {$y_4$};
  \node[X]         (I) [right of=P] {$x_9$};
  \node[X]         (J) [right of=I] {$x_{10}$};
  
  \node[t] [below=0.2cm of T] {$\sigma_{X \cup Y}(v) =$};
  \node[t] [below=0.2cm of A] {1};
  \node[t] [below=0.2cm of B] {2};
  \node[t] (x1) [below=0.2cm of M] {3};
  \node[t] [below=0.2cm of C] {4};
  \node[t] (x2) [below=0.2cm of N] {5};
  \node[t] [below=0.2cm of D] {6};
  \node[t] [below=0.2cm of E] {7};
  \node[t] [below=0.2cm of F] {8};
  \node[t] (x3) [below=0.2cm of O] {9};
  \node[t] [below=0.2cm of G] {10};
  \node[t] [below=0.2cm of H] {11};
  \node[t] (x4) [below=0.2cm of P] {12};
  \node[t] [below=0.2cm of I] {13};
  \node[t] [below=0.2cm of J] {14};

  \node[t] (t1) [below=0cm of x1] {\rotatebox[origin=c]{-90}{$=$}};
  \node[t] (t2) [below=0cm of x2] {\rotatebox[origin=c]{-90}{$=$}};
  \node[t] (t3) [below=0cm of x3] {\rotatebox[origin=c]{-90}{$=$}};
  \node[t] (t4) [below=0cm of x4] {\rotatebox[origin=c]{-90}{$=$}};
  
  \node[t] [below=0cm of t1] {$l_1$};
  \node[t] [below=0cm of t2] {$l_2$};
  \node[t] [below=0cm of t3] {$l_3$};
  \node[t] [below=0cm of t4] {$l_4$};
  
  \draw[decorate,decoration={brace,amplitude=3pt}] 
  (0.5,0.5) coordinate (t_k_unten) -- (2.5,0.5) coordinate (t_k_opt_unten); 
  \node[t] at (1.5,1) {$L_0$};
  \draw[decorate,decoration={brace,amplitude=3pt}] 
  (3.5,0.5) coordinate (t_k_unten) -- (4.5,0.5) coordinate (t_k_opt_unten); 
  \node[t] at (4,1) {$L_1$};
  \draw[decorate,decoration={brace,amplitude=3pt}] 
  (5.5,0.5) coordinate (t_k_unten) -- (8.5,0.5) coordinate (t_k_opt_unten); 
  \node[t] at (7,1) {$L_2$};
  \draw[decorate,decoration={brace,amplitude=3pt}] 
  (9.5,0.5) coordinate (t_k_unten) -- (11.5,0.5) coordinate (t_k_opt_unten); 
  \node[t] at (10.5,1) {$L_3$};
  \draw[decorate,decoration={brace,amplitude=3pt}] 
  (12.5,0.5) coordinate (t_k_unten) -- (14.5,0.5) coordinate (t_k_opt_unten); 
  \node[t] at (13.5,1) {$L_4$};
\end{tikzpicture}
\end{adjustbox}
	\caption{Visualization of \cref{defsL} 
	.}
	\label{ap-cbgdefsexample}
\end{figure}



\begin{definition}
	Let $\sigma_{X \cup Y} = \sigma_{L_0}\prod\limits_{i=1}^{|Y|}\sigma_{\{y_i\}}\sigma_{L_i}$ be an arbitrary ordering.
	Let us define
	$$shift_L(\sigma_{X \cup Y}) = \sigma_{\{y_{\lfloor \frac{|Y|}{2} \rfloor}\}}\sigma_{L_0}\Bigg(\prod\limits_{i=1}^{\lfloor\frac{|Y|}{2}\rfloor-1}\sigma_{\{y_i\}}\sigma_{L_i}\Bigg)\sigma_{L_{\lfloor \frac{|Y|}{2} \rfloor}}\Bigg(\prod\limits_{i=\lfloor\frac{|Y|}{2}\rfloor+1}^{|Y|}\sigma_{\{y_i\}}\sigma_{L_i}\Bigg),$$
	$$shift_R(\sigma_{X \cup Y}) =
	\sigma_{L_0}\Bigg(\prod\limits_{i=1}^{\lceil\frac{|Y|}{2}\rceil}\sigma_{\{y_i\}}\sigma_{L_i}\Bigg)\sigma_{L_{\lceil \frac{|Y|}{2} \rceil+1}}\Bigg(\prod\limits_{i=\lceil\frac{|Y|}{2}\rceil+2}^{|Y|}\sigma_{\{y_i\}}\sigma_{L_i}\Bigg)\sigma_{\{y_{\lceil \frac{|Y|}{2} \rceil+1}\}}.$$
	That is, $shift_L$ moves vertex $y_{\lfloor \frac{|Y|}{2} \rfloor}$ to the left most position and $shift_R$ moves vertex $y_{\lceil \frac{|Y|}{2} \rceil+1}$ to the right most position.
\end{definition}
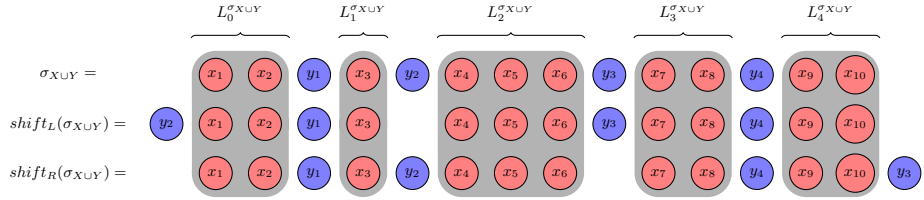
\begin{figure}[H]
	\centering
	\begin{adjustbox}{width=\textwidth}
\begin{tikzpicture}[-,semithick]
  
  \tikzset{Y/.append style={fill=blue!50,draw=black,text=black,shape=circle}}
  \tikzset{X/.append style={fill=red!50,draw=black,text=black,shape=circle}}
  \tikzset{t/.append style={fill=white,draw=white,text=black,align=right}}
  \node[t]         (T) {$\sigma_{X \cup Y} =$};
  \node[]         (empty1) [right of=T] {$~$};
  \node[]         (empty2) [right of=empty1] {$~$};
  \node[X]         (A) [right of=empty2] {$x_1$};
  \node[X]         (B) [right of=A] {$x_2$};
  \node[Y]         (M) [right of=B] {$y_1$};
  \node[X]         (C) [right of=M] {$x_3$};
  \node[Y]         (N) [right of=C] {$y_2$};
  \node[X]         (D) [right of=N] {$x_4$};
  \node[X]         (E) [right of=D] {$x_5$};
  \node[X]         (F) [right of=E] {$x_6$};
  \node[Y]         (O) [right of=F] {$y_3$};
  \node[X]         (G) [right of=O] {$x_7$};
  \node[X]         (H) [right of=G] {$x_8$};
  \node[Y]         (P) [right of=H] {$y_4$};
  \node[X]         (I) [right of=P] {$x_9$};
  \node[X]         (J) [right of=I] {$x_{10}$};

  \node[t]         (LT) [below of=T] {$shift_L(\sigma_{X \cup Y}) =$};
  \node[]         (Lempty1) [right of=LT] {$~$};
  \node[Y]         (LN) [right of=Lempty1] {$y_2$};
 \node[X]         (LA) [right of=LN] {$x_1$};
 \node[X]         (LB) [right of=LA] {$x_2$};
 \node[Y]         (LM) [right of=LB] {$y_1$};
 \node[X]         (LC) [right of=LM] {$x_3$};
 \node[]         (Lempty) [right of=LC] {$~$};
 \node[X]         (LD) [right of=Lempty] {$x_4$};
 \node[X]         (LE) [right of=LD] {$x_5$};
 \node[X]         (LF) [right of=LE] {$x_6$};
 \node[Y]         (LO) [right of=LF] {$y_3$};
 \node[X]         (LG) [right of=LO] {$x_7$};
 \node[X]         (LH) [right of=LG] {$x_8$};
 \node[Y]         (LP) [right of=LH] {$y_4$};
 \node[X]         (LI) [right of=LP] {$x_9$};
 \node[X]         (LJ) [right of=LI] {$x_{10}$};

  \node[t]         (RT) [below of=LT] {$shift_R(\sigma_{X \cup Y}) =$};
  \node[]         (Rempty1) [right of=RT] {$~$};
  \node[]         (Rempty2) [right of=Rempty1] {$~$};
  \node[X]         (RA) [right of=Rempty2] {$x_1$};
  \node[X]         (RB) [right of=RA] {$x_2$};
  \node[Y]         (RM) [right of=RB] {$y_1$};
  \node[X]         (RC) [right of=RM] {$x_3$};
  \node[Y]         (RN) [right of=RC] {$y_2$};
  \node[X]         (RD) [right of=RN] {$x_4$};
  \node[X]         (RE) [right of=RD] {$x_5$};
  \node[X]         (RF) [right of=RE] {$x_6$};
  \node[]         (Rempty2) [right of=RF] {$~$};
  \node[X]         (RG) [right of=Rempty2] {$x_7$};
  \node[X]         (RH) [right of=RG] {$x_8$};
  \node[Y]         (RP) [right of=RH] {$y_4$};
  \node[X]         (RI) [right of=RP] {$x_9$};
  \node[X]         (RJ) [right of=RI] {$x_{10}$};
  \node[Y]         (RO) [right of=RJ] {$y_3$};
  
  \draw[decorate,decoration={brace,amplitude=3pt}] 
  (2.5,0.75) coordinate (t_k_unten) -- (4.5,0.75) coordinate (t_k_opt_unten); 
  \node[t] at (3.5,1.25) {$L_0^{\sigma_{X \cup Y}}$};
  \draw[decorate,decoration={brace,amplitude=3pt}] 
  (5.5,0.75) coordinate (t_k_unten) -- (6.5,0.75) coordinate (t_k_opt_unten); 
  \node[t] at (6,1.25) {$L_1^{\sigma_{X \cup Y}}$};
  \draw[decorate,decoration={brace,amplitude=3pt}] 
  (7.5,0.75) coordinate (t_k_unten) -- (10.5,0.75) coordinate (t_k_opt_unten); 
  \node[t] at (9,1.25) {$L_2^{\sigma_{X \cup Y}}$};
  \draw[decorate,decoration={brace,amplitude=3pt}] 
  (11.5,0.75) coordinate (t_k_unten) -- (13.5,0.75) coordinate (t_k_opt_unten); 
  \node[t] at (12.5,1.25) {$L_3^{\sigma_{X \cup Y}}$};
  \draw[decorate,decoration={brace,amplitude=3pt}] 
  (14.5,0.75) coordinate (t_k_unten) -- (16.5,0.75) coordinate (t_k_opt_unten); 
  \node[t] at (15.5,1.25) {$L_4^{\sigma_{X \cup Y}}$};
  
  \begin{scope}[on background layer]
  \fill[black,opacity=0.3] \convexpath{A,B,LB,RB,RA,LA}{1.5em}; 
  \fill[black,opacity=0.3] \convexpath{C,LC,RC}{1.5em};  
  \fill[black,opacity=0.3] \convexpath{D,F,RF,RD}{1.5em}; 
  \fill[black,opacity=0.3] \convexpath{G,H,RH,RG}{1.5em}; 
  \fill[black,opacity=0.3] \convexpath{I,J,RJ,RI}{1.5em}; 
  \end{scope}
\end{tikzpicture}
\end{adjustbox}
	\caption{Visualization of the $shift_L$ and $shift_R$ functions with an ordering $\sigma_{X \cup Y}$ on the vertices of graph $G = K_{10,4}$.}
	\label{ap-shift}
\end{figure}
\begin{lemma}\label{lem2}
	Let $\sigma_{X \cup Y}$ be an arbitrary imbalance optimal ordering. Let $\sigma'_{X \cup Y} = shift_L(\sigma_{X \cup Y})$. We have that $I(\sigma_{X \cup Y}) = I(\sigma'_{X \cup Y})$.
\end{lemma}

\begin{proof}
	We have that $I(\sigma_{X \cup Y}, y_{\lfloor \frac{|Y|}{2} \rfloor}) = \Bigg|\Bigg(\sum\limits_{i=0}^{\lfloor \frac{|Y|}{2} \rfloor-1}|L^{\sigma_{X \cup Y}}_i|\Bigg) - \Bigg(\sum\limits_{i=\lfloor \frac{|Y|}{2} \rfloor}^{|Y|}|L^{\sigma_{X \cup Y}}_i|\Bigg)\Bigg|$ and $I(\sigma'_{X \cup Y}, y_{\lfloor \frac{|Y|}{2} \rfloor}) = \sum\limits_{i=0}^{|Y|}|L^{\sigma_{X \cup Y}}_i| = |X|$. For each $0 \leq i \leq \lfloor \frac{|Y|}{2} \rfloor-1$, the imbalance of all the vertices in $L^{\sigma_{X \cup Y}}_i$ is smaller by 2 in $\sigma'_{X \cup Y}$. The imbalance of the remaining vertices remain the same. 
	\begin{case} $\Bigg(\sum\limits_{i=0}^{\lfloor \frac{|Y|}{2} \rfloor-1}|L^{\sigma_{X \cup Y}}_i|\Bigg) > \Bigg(\sum\limits_{i=\lfloor \frac{|Y|}{2} \rfloor}^{|Y|}|L^{\sigma_{X \cup Y}}_i|\Bigg)$.\\
	We can express the imbalance of $I(\sigma'_{X \cup Y})$ as follows:
	\begin{align*}
		I(\sigma'_{X \cup Y})  & = I(\sigma_{X \cup Y}) - 2\cdot\Bigg( \sum\limits_{i=0}^{\lfloor \frac{|Y|}{2} \rfloor-1}|L^{\sigma_{X \cup Y}}_i| \Bigg) - I(\sigma_{X \cup Y}, y_{\lfloor \frac{|Y|}{2} \rfloor}) \\
		& + I(\sigma'_{X \cup Y}, y_{\lfloor \frac{|Y|}{2} \rfloor})\\
		& = I(\sigma_{X \cup Y}) - 2\cdot\Bigg( \sum\limits_{i=0}^{\lfloor \frac{|Y|}{2} \rfloor-1}|L^{\sigma_{X \cup Y}}_i| \Bigg) - \Bigg(\sum\limits_{i=0}^{\lfloor \frac{|Y|}{2} \rfloor-1}|L^{\sigma_{X \cup Y}}_i|\Bigg) \\
		&+ \Bigg(\sum\limits_{i=\lfloor \frac{|Y|}{2} \rfloor}^{|Y|}|L^{\sigma_{X \cup Y}}_i|\Bigg) + \sum\limits_{i=0}^{|Y|}|L^{\sigma_{X \cup Y}}_i|\\
		& <  I(\sigma_{X \cup Y}).
	\end{align*}
	Since $\sigma_{X \cup Y}$ it is not possible that $I(\sigma'_{X \cup Y}) < I(\sigma_{X \cup Y})$. Thus this case can not occur.
	\end{case}
	\begin{case}
		$\Bigg(\sum\limits_{i=0}^{\lfloor \frac{|Y|}{2} \rfloor-1}|L^{\sigma_{X \cup Y}}_i|\Bigg) \leq \Bigg(\sum\limits_{i=\lfloor \frac{|Y|}{2} \rfloor}^{|Y|}|L^{\sigma_{X \cup Y}}_i|\Bigg)$.\\
		We can express the imbalance of $I(\sigma'_{X \cup Y})$ as follows:
		\begin{align*}
		I(\sigma'_{X \cup Y})  & = I(\sigma_{X \cup Y}) - 2\cdot\Bigg( \sum\limits_{i=0}^{\lfloor \frac{|Y|}{2} \rfloor-1}|L^{\sigma_{X \cup Y}}_i| \Bigg) - I(\sigma_{X \cup Y}, y_{\lfloor \frac{|Y|}{2} \rfloor}) \\
		& + I(\sigma'_{X \cup Y}, y_{\lfloor \frac{|Y|}{2} \rfloor})\\
		& = I(\sigma_{X \cup Y}) - 2\cdot\Bigg( \sum\limits_{i=0}^{\lfloor \frac{|Y|}{2} \rfloor-1}|L^{\sigma_{X \cup Y}}_i| \Bigg) - \Bigg(\sum\limits_{i=\lfloor \frac{|Y|}{2} \rfloor}^{|Y|}|L^{\sigma_{X \cup Y}}_i|\Bigg) \\
		&+ \Bigg(\sum\limits_{i=0}^{\lfloor \frac{|Y|}{2} \rfloor-1}|L^{\sigma_{X \cup Y}}_i|\Bigg)  + \sum\limits_{i=0}^{|Y|}|L_i|\\
		& =  I(\sigma_{X \cup Y}).
		\end{align*}
	\end{case}
\end{proof}
\begin{lemma}\label{lem3}
	Let $\sigma_{X \cup Y}$ be an arbitrary imbalance optimal ordering. Let $\sigma'_{X \cup Y} = shift_R(\sigma_{X \cup Y})$. We have that $I(\sigma_{X \cup Y}) = I(\sigma'_{X \cup Y})$.
\end{lemma}

\begin{proof}
	Analogous to \cref{lem2}.
\end{proof}

\begin{remark}\label{rmrk4}
	Let both $|X|$ and $|Y|$ be odd. Let $\sigma_{X \cup Y}$ be an arbitrary imbalance optimal ordering. Let $y_m \in Y$ be the vertex at position $l_{\lceil\frac{|Y|}{2}\rceil}$.
	We have that $I(\sigma_{X \cup Y}, y_m) = 1$. Otherwise $\sigma_{X \cup Y}$ is not imbalance optimal. This can be shown by a proof by contradiction using \cref{lem2} and \cref{lem3}.
\end{remark}

\begin{theorem}\label{thm1}
	If $G = (X,Y,E)$ is a complete bigraph, then the minimum imbalance of $G$ is
	$|X|\cdot |Y| + (|X| \mod 2) \cdot (|Y| \mod 2).$
\end{theorem}

\begin{proof}
Let $\sigma_{X \cup Y}$ be an arbitrary imbalance optimal ordering. Repeatedly apply the functions $shift_L$ and $shift_R$ on $\sigma_{X \cup Y}$, until we have the same ordering as constructed in \cref{lmacbg1}. Let us denote the obtained ordering as $\sigma'_{X \cup Y}$. Since $I(\sigma_{X \cup Y}) = I(\sigma'_{X \cup Y})$ by \cref{lem2} and \cref{lem3}, and $I(\sigma'_{X \cup Y}) = |X|\cdot |Y| + (|X| \mod 2) \cdot (|Y| \mod 2)$, we have that 
$I(G) = |X|\cdot |Y| + (|X| \mod 2) \cdot (|Y| \mod 2).$
\end{proof}

\begin{corollary}\label{col-prop}
	By \cref{lem2}, \cref{lem3}, and \cref{rmrk4}, any ordering $\sigma_{X \cup Y}$ of a complete bigraph $G=(X,Y,E)$ is imbalance optimal if and only if $\sigma_{X \cup Y}$ has the following 3 properties:
	\begin{enumerate}
		\item $\Bigg(\sum\limits_{i=0}^{\lfloor \frac{|Y|}{2} \rfloor-1}|L^{\sigma_{X \cup Y}}_i|\Bigg) \leq \Bigg(\sum\limits_{i=\lfloor \frac{|Y|}{2} \rfloor}^{|Y|}|L^{\sigma_{X \cup Y}}_i|\Bigg)$
		\item $\Bigg(\sum\limits_{i=0}^{\lceil \frac{|Y|}{2} \rceil}|L^{\sigma_{X \cup Y}}_i|\Bigg) \geq \Bigg(\sum\limits_{i=\lceil \frac{|Y|}{2} \rceil+1}^{|Y|}|L^{\sigma_{X \cup Y}}_i|\Bigg)$
		\item $|X|$ and $|Y|$ are odd $\implies$ $I(\sigma_{X \cup Y}, y_m) = 1$, where $y_m \in Y$ is the vertex at position $l_{\lceil \frac{|Y|}{2} \rceil}$.
	\end{enumerate}
These properties allow us to verify whether any ordering is imbalance optimal in $O(|X|+|Y|)$.
\end{corollary}

\begin{corollary}\label{thm1.5}
	Let $G = (X,Y,E)$ be a complete bigraph and let $|X| + |Y| = n$. Given $|X|$ and $|Y|$, the minimum imbalance $I(G)$ can be computed in $\mathcal{O}(\log(n) \cdot \log(\log(n)))$ time by using the formula of \cref{thm1}.
	This follows from the fact that the product of two $k$-bit integers can be computed in $\mathcal{O}(k \cdot \log(k))$ time\cite{10.4007/annals.2021.193.2.4}. 
\end{corollary}
\section{Imbalance on Chained Complete bipartite graphs}\label{c4}
In this section we shall introduce the chained complete bigraph. We show that the chained complete bigraph is a subclass of PI-bigraphs and how to use the results of \cref{c3} to compute its minimum imbalance efficiently. 

\begin{definition}
	We define $\mathscr{C}$ to be a family of maximal subsets of the vertices of graph $G = (V,E)$ that induce a complete bigraph on $G$. Additionally, for all edges $e \in E$ there exists a vertex set $C_i \in \mathscr{C}$ such that both endpoints of $e$ are contained in $C_i$. That is, $\mathscr{C} \subseteq \mathcal{P}(V)$ such that:
	\begin{align*}
	(\forall_{\{u,v\}\in E}~\exists_{C_i\in \mathscr{C}}~ \{u,v\} \subseteq C_i) ~\wedge 
	(\forall_{C_i \in \mathscr{C}}~\exists_{j,k \in \mathbb{N}}~ G[C_i] = K_{j,k}) ~\wedge \\ 
	(\forall_{C_i \in \mathscr{C}}~\forall_{v \in V \setminus C_i}~\forall_{j,k \in \mathbb{N}}~ G[C_i \cup \{v\}] \neq K_{j,k}),
	\end{align*}
		where $K_{i,j}$ denotes a complete bigraph. We call $\mathscr{C}$ the maximal complete bigraph components, abbreviated by MCB-components, of $G$.
\end{definition}

\begin{definition}\label{restpibigraphdef}
	A graph $G = (X,Y,E)$ is a chained complete bigraph, if $G$ has a MCB-component family $\mathscr{C}$ such that we can label 
	$\mathscr{C}=\{C_1, \dots, C_n\}$
	such that consecutive vertex sets $C_i$ and $C_{i+1}$ share exactly one vertex and non-consecutive vertex sets $C_i$ and $C_j$ share no vertices. Formally,
	$(\forall_{1\leq i < n}~|C_i \cap C_{i+1}| = 1) \wedge (\forall_{1\leq i < j < n}~j-i > 1 \implies C_i \cap C_j = \emptyset).$
	We call a vertex that is shared by two consecutive vertex sets of $\mathscr{C}$ an overlapping vertex.
\end{definition}

For any chained complete bigraph $G$, with corresponding MCB-component family $\mathscr{C}$, we can create a corresponding PI-bigraph interval representation $I_G$. For each $C_i \in \mathscr{C}$ we create a staircase-shaped set of intervals with the overlapping vertices at the top and bottom.

\begin{figure}[H]
	\centering
	\begin{tikzpicture}[-,semithick]
  
  \tikzset{Y/.append style={fill=blue!50,draw=black,text=black,shape=circle}}
  \node[Y]         (y1) at (0,0) {$y_1$};
  \node[Y]         (y2) at (2,0) {$y_2$};
  \node[Y]         (y3) at (4,0) {$y_3$};
  \node[Y]         (y4) at (6,0) {$y_4$};
  \node[Y]         (y5) at (8,0) {$y_5$};
  \node[Y]         (y6) at (10,0) {$y_6$};
  
  \tikzset{X/.append style={fill=red!50,draw=black,text=black,shape=circle}}
  \node[X]         (x1) at (1,1) {$x_1$};
  \node[X]         (x2) at (3,1) {$x_2$};
  \node[X]         (x3) at (5,1) {$x_3$};
  \node[X]         (x4) at (7,1) {$x_4$};
  \node[X]         (x5) at (9,1) {$x_5$};

  \tikzset{t/.append style={fill=white,draw=white,text=black}}
  \node[t]         at (1,-1) {$C_1$};
  \node[t]         at (5,-1) {$C_2$};
  \node[t]         at (9,-1) {$C_3$};

  \path (x1) edge              node {} (y1)
             edge              node {} (y2)
        (x2) edge              node {} (y1)
             edge              node {} (y2)
             edge              node {} (y3)
             edge              node {} (y4)
        (x3) edge              node {} (y3)
             edge              node {} (y4)
        (x4) edge              node {} (y5)
             edge              node {} (y6)
             edge              node {} (y4)
        (x5) edge              node {} (y5)
             edge              node {} (y6)
             edge              node {} (y4);
   \begin{scope}[on background layer]     
   \fill[black,opacity=0.3] \convexpath{x1,x2,y2,y1}{1.5em}; 
   \fill[black,opacity=0.3] \convexpath{x2,x3,y4,y3}{1.5em}; 
   \fill[black,opacity=0.3] \convexpath{x4,x5,y6,y5,y4}{1.5em};
   \end{scope}
\end{tikzpicture}
	\caption{Example of a chained complete bigraph $G=(X,Y,E)$, where $X$ is represented by the red vertices $x_i$ and $Y$ by the blue vertices $y_i$. The highlighted areas represent the vertex sets of $\mathscr{C}$.}
	\label{ap-rpibg}
\end{figure}
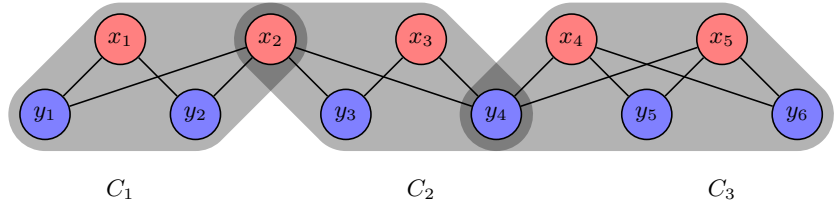
\begin{figure}[H]
	\centering
	\begin{adjustbox}{scale=1}
\begin{tikzpicture}
    \node[draw,minimum height=0.0cm,minimum width=4cm, fill=red!50] at 
     (2.3,-0.628) {$x_2$};
     
    \node[draw,minimum height=0.0cm,minimum width=2cm, fill=blue!50] at 
     (1,-1.1) {$y_2$};
    \node[draw,minimum height=0.0cm,minimum width=2cm, fill=blue!50] at 
     (0.5,-1.5725) {$y_1$};
    \node[draw,minimum height=0.0cm,minimum width=2cm, fill=red!50] at 
     (0,-2.04) {$x_1$};

    \node[draw,minimum height=0.0cm,minimum width=2cm, fill=blue!50] at 
     (3.5,-1.1) {$y_3$};
    \node[draw,minimum height=0.0cm,minimum width=2cm, fill=red!50] at 
     (4.,-1.5725) {$x_3$};
    \node[draw,minimum height=0.0cm,minimum width=4cm, fill=blue!50] at 
     (5.25,-2.04) {$y_4$};
     
    \node[draw,minimum height=0.0cm,minimum width=2cm, fill=red!50] at 
     (7.5,-0.175) {$x_5$};
    \node[draw,minimum height=0.0cm,minimum width=2cm, fill=red!50] at 
     (7.25,-0.628) {$x_4$};
    \node[draw,minimum height=0.0cm,minimum width=2cm, fill=blue!50] at 
     (7.,-1.1) {$y_6$};
    \node[draw,minimum height=0.0cm,minimum width=2cm, fill=blue!50] at 
     (6.75,-1.5725) {$y_5$};
     
     \draw[->] (-1,-2.25) -- (9,-2.25) node[right] {$\mathbb{R}$};
     
     \draw [fill=black!10, opacity=0.3] (2.25,-2.9) rectangle (-1,0.1);
     \draw [fill=red!10, opacity=0.3] (5.25,-2.9) rectangle (2.25,0.1);
     \draw [fill=black!10, opacity=0.3] (8.5,-2.9) rectangle (5.25,0.1);
     
     \node at (0.75,-2.6) {$C_1$};
     \node at (3.75,-2.6) {$C_2$};
     \node at (6.75,-2.6) {$C_3$};
    \end{tikzpicture}
\end{adjustbox}
	\caption{Interval representation of the chained complete bigraph of \cref{ap-rpibg}.}
	\label{ap-rpibg2}
\end{figure}
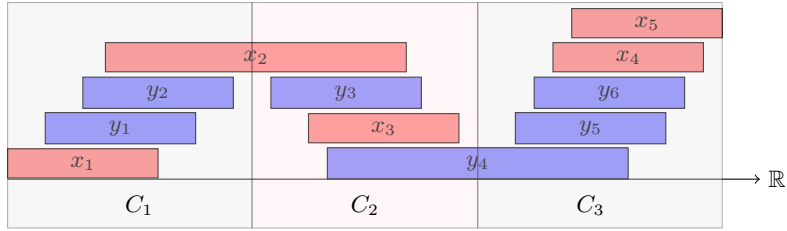

\begin{remark}
	By the definition of chained complete bigraph $G=(X,Y,E)$ we have 
	$\forall_{v \in X\cup Y}~ 1 \leq |\{C_i \in \mathscr{C} ~|~ v \in C_i\}| \leq 2.$
\end{remark}

\begin{remark}\label{astrtrip}
	By the definition of MCB-components $\mathscr{C}$ of a chained complete bigraph 
	$\forall_{v \in X \cup Y}~ N(v) \subseteq \bigcup\limits_{C_j \in \{C_i \in \mathscr{C} ~|~ v \in C_i\}}C_j.$
\end{remark}

Let $G = (X,Y,E)$ be a chained complete bigraph with corresponding MCB-component family $\mathscr{C} = \{C_1, \dots, C_n\}$. Then we have that:
\begin{align*}
I(G) &= \sum\limits_{i=1}^{n}|X_i|\cdot|Y_i| + (|X_i| \mod 2) \cdot (|Y_i|\mod 2) \\
&- \Bigg(\sum\limits_{i=1}^{n-1}g(s_i,C_i) + g(s_i,C_{i+1})\Bigg) + \Bigg(\sum\limits_{i=1}^{n-1}|g(s_i,C_i) - g(s_i,C_{i+1})|\Bigg),
\end{align*}
where $X_i$, $Y_i$, $s_i$, and the function $g$ are defined below. 
First, we introduce additional definitions that are required to understand the proof.
\begin{definition}
	Let $G = (X, Y, E)$ be a chained complete bigraph with corresponding MCB-component family $\mathscr{C} = \{C_1, \dots, C_n\}$. We shall use the notation $G[C_i] = (X_i, Y_i, E_i)$ to denote the graph induced on $C_i \in \mathscr{C}.$
\end{definition}
\begin{definition}
	Let $G$ be a chained complete bigraph with corresponding MCB-component family $\mathscr{C} = \{C_1, \dots, C_n\}$. We label the overlapping vertices of $\mathscr{C}$ as $s_i \in C_i \cap C_{i+1},$
	where $1 \leq i \leq n-1$. By the definition of chained complete bigraph, the overlapping vertex $s_i$ is unique. We define $S = \{s_i ~|~ 1 \leq i \leq n-1\}.$
\end{definition}

\begin{definition}
	We define $g(s_i, C_j)$, where $s_i \in S$ and $C_j \in \mathscr{C}$, to be the number of neighbors of $s_i$ in $C_j$. Equivalently, $g(s_i, C_j)$ is the number of vertices in $C_j$ that do not belong to the same part as $s_i$. That is,
	$g(s_i, C_j) = |N(s_i) \cap C_j| = \begin{cases}
	|X_j| & \text{if } s_i \in Y\\
	|Y_j| & \text{if } s_i \in X
	\end{cases}.$
\end{definition}
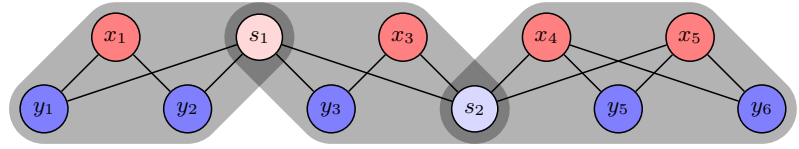
\begin{figure}[H]
	\centering
	\resizebox{\textwidth}{!}{%
\begin{tikzpicture}[-,semithick]
  \tikzset{XS/.append style={fill=red!15,draw=black,text=black,shape=circle}}
  \node[XS]         (x2) at (1,1) {$s_1$};
  
  \tikzset{YS/.append style={fill=blue!15,draw=black,text=black,shape=circle}}
  \node[YS]         (y4) at (4,0) {$s_2$};
  
  \tikzset{Y/.append style={fill=blue!50,draw=black,text=black,shape=circle}}
  \node[Y]         (y1) at (-2,0) {$y_1$};
  \node[Y]         (y2) at (0,0) {$y_2$};
  \node[Y]         (y3) at (2,0) {$y_3$};
  \node[Y]         (y5) at (6,0) {$y_5$};
  \node[Y]         (y6) at (8,0) {$y_6$};
  
  \tikzset{X/.append style={fill=red!50,draw=black,text=black,shape=circle}}
  \node[X]         (x1) at (-1,1) {$x_1$};
  \node[X]         (x3) at (3,1) {$x_3$};
  \node[X]         (x4) at (5,1) {$x_4$};
  \node[X]         (x5) at (7,1) {$x_5$};

  \tikzset{t/.append style={fill=white,draw=white,text=black}}
  \node[t]         at (-1,-1) {$C_1$};
  \node[t]         at (3,-1) {$C_2$};
  \node[t]         at (7,-1) {$C_3$};
  \node[t]         at (0,-1.5) {$G[C_1] = (X_1,Y_1, E_1) = (\{x_1,s_1\}, \{y_1,y_2\}, X_1 \times Y_1)$};
  \node[t]         at (0,-2) {$G[C_2] = (X_2,Y_2, E_2) = (\{s_1,x_3\}, \{y_3,s_2\}, X_2 \times Y_2)$};
  \node[t]         at (0.3,-2.5) {$G[C_3] = (X_3,Y_3, E_3) = (\{x_4,x_5\}, \{s_2,y_5,y_6\}, X_3 \times Y_3)$};
  \node[t]         at (7.5,-1.5) {$g(s_1, C_1)=|Y_1|=2$};
  \node[t]         at (7.5,-2) {$g(s_1, C_2)=|Y_2|=2$};
  \node[t]         at (7.5,-2.5) {$g(s_2, C_2)=|X_2|=2$};
  \node[t]         at (7.5,-3) {$g(s_2, C_3)=|X_3|=2$};

  \path (x1) edge              node {} (y1)
             edge              node {} (y2)
        (x2) edge              node {} (y1)
             edge              node {} (y2)
             edge              node {} (y3)
             edge              node {} (y4)
        (x3) edge              node {} (y3)
             edge              node {} (y4)
        (x4) edge              node {} (y5)
             edge              node {} (y6)
             edge              node {} (y4)
        (x5) edge              node {} (y5)
             edge              node {} (y6)
             edge              node {} (y4);
   \begin{scope}[on background layer]     
   \fill[black,opacity=0.3] \convexpath{x1,x2,y2,y1}{1.5em}; 
   \fill[black,opacity=0.3] \convexpath{x2,x3,y4,y3}{1.5em}; 
   \fill[black,opacity=0.3] \convexpath{x4,x5,y6,y5,y4}{1.5em};
   \end{scope}
\end{tikzpicture}}%
	\caption{Illustration of the additional definitions of \cref{c4} on the example graph of \cref{ap-rpibg}.}
	\label{ap-additionaldefsrpig}
\end{figure}

\subsection{Proof of the upper bound}
\begin{lemma}\label{rpigimbaux1}
	Let $G = (X,Y,E)$ be a chained complete bigraph with corresponding MCB-component family $\mathscr{C} = \{C_1, \dots, C_n\}$. Let $C_i \in \mathscr{C}$ such that $|X_i| = 1$ or $|Y_i| = 1$. W.l.o.g. assume that $|X_i| = 1$, then neither overlapping vertices $s_{i-1}$ nor $s_i$ can be in $X_i$. Formally, $\forall_{C_i \in \mathscr{C}}(|X_i| = 1 \implies s_{i-1} \notin X_i \wedge s_{i} \notin X_i) ~\wedge (|Y_i| = 1 \implies s_{i-1} \notin Y_i \wedge s_{i} \notin Y_i).$
\end{lemma}
\begin{proof}
	Trivial proof by contradiction. (If not, then $C_i$ is not maximal.)
\end{proof}

\begin{lemma}\label{rpigimb1}
	Given a chained complete bigraph $G = (X,Y,E)$ with corresponding MCB-component family $\mathscr{C} = \{C_1, \dots, C_n\}$, we have that
	\begin{align*}
	I(G) &\leq \sum\limits_{i=1}^{n}|X_i|\cdot|Y_i| + (|X_i| \mod 2) \cdot (|Y_i|\mod 2) \\
	&- \Bigg(\sum\limits_{i=1}^{n-1}g(s_i,C_i) + g(s_i,C_{i+1})\Bigg) + \Bigg(\sum\limits_{i=1}^{n-1}|g(s_i,C_i) - g(s_i,C_{i+1})|\Bigg).
	\end{align*}
\end{lemma}

\begin{proof}
	The lemma is proven by constructing an ordering $\sigma_{X \cup Y}$ whose imbalance is equivalent to the above expression. The ordering $\sigma_{X \cup Y}$ is constructed by creating a subordering for each $C_i \in \mathscr{C}$ separately and concatenating those suborderings. The suborderings are created in a similar fashion as the orderings in the proof of \cref{lmacbg1}.
\end{proof}

\subsection{Proof of the lower bound}

\begin{remark}\label{ordanalysisrem3}
	The imbalance of $v \in (X \cup Y) \setminus S$ is only influenced by the vertices in $C_i$. That is,
	$\forall_{C_i \in \mathscr{C}} ~\forall_{v \in C_i \setminus S} ~I(v, \sigma_{X\cup Y}, G) = I(v, \sigma_{X\cup Y}^{C_i}, G[C_i]).$
\end{remark}
\begin{remark}\label{ordanalysisrem4}
	The imbalance of $s_i$ is only influenced by the vertices in $C_i \cup C_{i+1}$. That is,
	$\forall_{s_i \in S} ~I(s_{i}, \sigma_{X\cup Y}, G) = I(s_i, \sigma_{X\cup Y}^{C_{i} \cup C_{i+1}}, G[C_i \cup C_{i+1}]).$
\end{remark}

\begin{lemma}\label{rpigimbaux}
	Let $G=(X,Y,E)$ be a chained complete bigraph with corresponding MCB-component family $\mathscr{C} = \{C_1, \dots, C_n\}$. For any arbitrary ordering $\sigma_{X\cup Y}$ it holds that
	\begin{align*}
	\MoveEqLeft I(s_{n-1}, \sigma_{X \cup Y}^{C_{n-1} \cup C_n}, G[C_{n-1} \cup C_n])
		- I(s_{n-1}, \sigma_{X \cup Y}^{C_{n-1}}, G[C_{n-1}]) \\
	\MoveEqLeft	- I(s_{n-1}, \sigma_{X \cup Y}^{C_{n}}, G[C_{n}]) \\
	&\geq |g(s_{n-1},C_{n-1}) - g(s_{n-1},C_{n})| - g(s_{n-1},C_{n-1}) - g(s_{n-1},C_{n}).
	\end{align*}
\end{lemma}
\begin{proof}
	The expression
	$|g(s_{n-1},C_{n-1}) - g(s_{n-1},C_{n})| - g(s_{n-1},C_{n-1}) - g(s_{n-1},C_{n})$
	takes two possible values depending on the sign of $g(s_{n-1},C_{n-1}) - g(s_{n-1},C_{n})$. Either the above expression is equivalent to
	$-2 \cdot g(s_{n-1},C_{n})$
	or
	$-2 \cdot g(s_{n-1},C_{n-1}).$
	To relate the expressions in the inequality, we shall denote the number of neighbors of $s_{n-1}$ to its left and to its right in $C_n$ and $C_{n-1}$ in ordering $\sigma_{X \cup Y}$ as:
	
	\begin{tabular}{ l   }
		$l_1 = |\{v \in C_{n-1} \cap N(s_{n-1}) ~|~ v <_{\sigma_{X \cup Y}} s_{n-1}\}|$; \\ 
		$l_2 = |\{v \in C_{n} \cap N(s_{n-1}) ~|~ v <_{\sigma_{X \cup Y}} s_{n-1}\}|$; \\ 
		$r_1 = |\{v \in C_{n-1} \cap N(s_{n-1}) ~|~ v >_{\sigma_{X \cup Y}} s_{n-1}\}|$; \\ 
		$r_2 = |\{v \in C_{n} \cap N(s_{n-1}) ~|~ v >_{\sigma_{X \cup Y}} s_{n-1}\}|.$    
	\end{tabular}\\
	Using the above definitions, we rewrite the expression on the left-hand side of the inequality as follows:
	\begin{align*}
	\MoveEqLeft I(s_{n-1}, \sigma_{X \cup Y}^{C_{n-1} \cup C_n}, G[C_{n-1} \cup C_n]) - I(s_{n-1}, \sigma_{X \cup Y}^{C_{n-1}}, G[C_{n-1}]) \\
	\MoveEqLeft - I(s_{n-1}, \sigma_{X \cup Y}^{C_{n}}, G[C_{n}]) \\
	&= |l_1 - r_1 + l_2 - r_2| - |l_1 - r_1| - |l_2 - r_2|.
	\end{align*}
	According to the signs of $l_1 - r_1$ and $l_2 - r_2$, consider the following cases:
	
	\begin{tabular}{ l l l l }
		$(++)$ & $l_1 - r_1 \geq 0$ and $l_2 - r_2 \geq 0$; & $(+-)$ & $l_1 - r_1 \geq 0$ and $l_2 - r_2 < 0$; \\ 
		$(-+)$ & $l_1 - r_1 < 0$ and $l_2 - r_2 \geq 0$; & $(--)$ & $l_1 - r_1 < 0$ and $l_2 - r_2 < 0$.
	\end{tabular}\\
	For the cases $(++)$ and $(--)$, we have that
	$|l_1 - r_1 + l_2 - r_2| - |l_1 - r_1| - |l_2 - r_2| = 0.$
	Thus, in the cases $(++)$ and $(--)$, it holds that
	\begin{eqnarray*}
		\lefteqn{|l_1 - r_1 + l_2 - r_2| - |l_1 - r_1| - |l_2 - r_2| = 0} \\ 
		& \geq & |g(s_{n-1},C_{n-1}) - g(s_{n-1},C_{n})| - g(s_{n-1},C_{n-1}) - g(s_{n-1},C_{n}).
	\end{eqnarray*}
	This follows from the fact that $-2 \cdot g(s_{n-1},C_{n}) \leq 0$
	and
	$-2 \cdot g(s_{n-1},C_{n-1}) \leq 0.$
	For the cases $(+-)$ and $(-+)$, we have that
	$$
	|l_1 - r_1 + l_2 - r_2| - |l_1 - r_1| - |l_2 - r_2|=
	\begin{cases}
	2 (l_2 - r_2) & \text{if } (+-) \wedge l_1 - r_1 + l_2 - r_2 \geq 0 \\
	2 (r_2 - l_2) & \text{if } (-+) \wedge l_1 - r_1 + l_2 - r_2 < 0 \\
	2 (r_1 - l_1) & \text{if } (+-) \wedge l_1 - r_1 + l_2 - r_2 < 0 \\
	2 (l_1 - r_1) & \text{if } (-+) \wedge l_1 - r_1 + l_2 - r_2 \geq 0 
	\end{cases}.
	$$
	Observe that, by the definitions of $l_1$, $r_1$, $l_2$, $r_2$, and function $g$, we have\\
	$l_1 + r_1 = g(s_{n-1}, C_{n-1})$
	and 
	$l_2 + r_2 = g(s_{n-1}, C_{n}).$
	Using the above remark and case distinction, we derive that 
	in the cases $(+-)$ and $(-+)$ it 
	holds that 
	\begin{align*}
	\lefteqn{|l_1 - r_1 + l_2 - r_2| - |l_1 - r_1| - |l_2 - r_2|}\\
	& \geq |g(s_{n-1},C_{n-1}) - g(s_{n-1},C_{n})| - g(s_{n-1},C_{n-1}) - g(s_{n-1},C_{n}). 
	\end{align*}
\end{proof}
\begin{lemma}\label{rpigimb2}
	Given a chained complete bigraph $G = (X,Y,E)$ with corresponding MCB-component family $\mathscr{C} = \{C_1, \dots, C_n\}$, we have that
	\begin{align*}
	I(G) &\geq \sum\limits_{i=1}^{n}|X_i|\cdot|Y_i| + (|X_i| \mod 2) \cdot (|Y_i|\mod 2) \\
	&- \Bigg(\sum\limits_{i=1}^{n-1}g(s_i,C_i) + g(s_i,C_{i+1})\Bigg) + \Bigg(\sum\limits_{i=1}^{n-1}|g(s_i,C_i) - g(s_i,C_{i+1})|\Bigg).
	\end{align*}
\end{lemma}
\begin{proof}
	We shall prove that the imbalance of any arbitrary ordering $\sigma_{X \cup Y}$ on the vertex set $X \cup Y$ is bounded from below by the above expression by induction on $|\mathscr{C}| = n$.
	\begin{itemize}[label=$\bullet$]
		\item Base Case $(n = 0 \vee n =1)$:\\
		By the definition of MCB-components $\mathscr{C}$, the graph $G$ is an empty graph or a complete bigraph. Thus, by \cref{thm1}, the lemma holds for the base case.
		\item Induction step $(n>1)$:\\
		Let $G = (X,Y,E)$ be a chained complete bigraph with corresponding MCB-component family $\mathscr{C} = \{C_1, \dots, C_{k+1}\}$. Let us define $$\mathscr{C} \setminus C_{k+1} = \bigcup\limits_{C_i \in \mathscr{C}\setminus \{C_{k+1}\}}C_i.$$
		We write the imbalance of $\sigma_{X \cup Y}$ as follows:
		\begin{align}
		I(\sigma_{X \cup Y}) & = I(\sigma_{X\cup Y}^{\mathscr{C} \setminus C_{k+1}}, G[\mathscr{C} \setminus C_{k+1}]) + I(\sigma_{X\cup Y}^{C_{k+1}}, G[C_{k+1}])\notag\\
		& - I(s_{k}, \sigma_{X \cup Y}^{C_{k}}, G[C_{k}]) - I(s_{k}, \sigma_{X \cup Y}^{C_{k+1}}, G[C_{k+1}]) \notag\\
		&+ I(s_{k}, \sigma_{X \cup Y}^{C_{k} \cup C_{k+1}}, G[C_{k} \cup C_{k+1}]) \label{eqlem7-0}\\
		&\geq\Bigg(\sum\limits_{i=1}^{k}|X_i|\cdot|Y_i| + (|X_i| \mod 2) \cdot (|Y_i|\mod 2)\Bigg)\notag\\
		&- \Bigg(\sum\limits_{i=1}^{k-1}g(s_i,C_i) + g(s_i,C_{i+1})\Bigg) + \Bigg(\sum\limits_{i=1}^{k-1}|g(s_i,C_i) - g(s_i,C_{i+1})|\Bigg) \notag\\
		&+ |X_{k+1}|\cdot|Y_{k+1}| + (|X_{k+1}| \mod 2) \cdot (|Y_{k+1}|\mod 2) \notag\\
		&- I(s_{k}, \sigma_{X \cup Y}^{C_{k}}, G[C_{k}])  - I(s_{k}, \sigma_{X \cup Y}^{C_{k+1}}, G[C_{k+1}]) \notag\\ 
		&+ I(s_{k}, \sigma_{X \cup Y}^{C_{k} \cup C_{k+1}}, G[C_{k} \cup C_{k+1}]) \label{eqlem7-1}\\
		& \geq\Bigg(\sum\limits_{i=1}^{k+1}|X_i|\cdot|Y_i| + (|X_i| \mod 2) \cdot (|Y_i|\mod 2)\Bigg) \notag\\
		& - \Bigg(\sum\limits_{i=1}^{k-1}g(s_i,C_i) + g(s_i,C_{i+1})\Bigg) + \Bigg(\sum\limits_{i=1}^{k-1}|g(s_i,C_i) - g(s_i,C_{i+1})|\Bigg)\notag\\
		&- g(s_{k}, C_{k+1}) - g(s_{k}, C_{k}) + |g(s_{k}, C_{k}) - g(s_{k}, C_{k+1})| \label{eqlem7-2}\\
		&= \Bigg(\sum\limits_{i=1}^{k+1}|X_i|\cdot|Y_i| + (|X_i| \mod 2) \cdot (|Y_i|\mod 2)\Bigg)\notag\\
		&- \Bigg(\sum\limits_{i=1}^{k}g(s_i,C_i) + g(s_i,C_{i+1})\Bigg) + \Bigg(\sum\limits_{i=1}^{k}|g(s_i,C_i) - g(s_i,C_{i+1})|\Bigg), \notag
		\end{align}
		where \cref{eqlem7-0} follows from \cref{ordanalysisrem3} and \cref{ordanalysisrem4}, \cref{eqlem7-1} follows from the induction hypothesis, and \cref{eqlem7-2} follows from \cref{rpigimbaux}. 
	\end{itemize}
\end{proof}
\begin{theorem}\label{thm2}
	Let $G = (X,Y,E)$ be a chained complete bigraph with corresponding MCB-component family $\mathscr{C} = \{C_1, \dots, C_n\}$. We have that
\end{theorem}
\begin{align*}
I(G) &= \sum\limits_{i=1}^{n}|X_i|\cdot|Y_i| + (|X_i| \mod 2) \cdot (|Y_i|\mod 2) \\
&- \Bigg(\sum\limits_{i=1}^{n-1}g(s_i,C_i) + g(s_i,C_{i+1})\Bigg) + \Bigg(\sum\limits_{i=1}^{n-1}|g(s_i,C_i) - g(s_i,C_{i+1})|\Bigg).\\
\end{align*}
\begin{proof}
	Follows from \cref{rpigimb1} and \cref{rpigimb2}.
\end{proof}
\begin{corollary}\label{thm4}
	Let $G = (X,Y,E)$ be a chained complete bigraph with corresponding MCB-component family $\mathscr{C} = \{C_1, \dots, C_n\}$ and let $|X| + |Y| = m$. Given $\{|X_1|, \dots, |X_n|\}$, $\{|Y_1|, \dots, |Y_n|\}$, and $\{s_1 \in X, \dots, s_{n-1} \in X\}$,
	the imbalance of $G$ can be computed in $\mathcal{O}(n \cdot \log(m) \cdot \log(\log(m)))$ time. By applying a similar reasoning as in \cref{thm1.5}, we verify the correctness of this corollary.
\end{corollary}


\bibliographystyle{splncs04}
\bibliography{references} 
\newpage
\section{Appendix - Full Proof}

\setcounter{equation}{0}
\setcounter{subsection}{3}
\subsection{Imbalance on Chained Complete bipartite graphs}

\subsubsection{Proof of the upper bound}

\begin{lemma}\label{ap-lem4}
	Let $G = (X,Y,E)$ be a chained complete bigraph with corresponding MCB-component family $\mathscr{C} = \{C_1, \dots, C_n\}$. Let $C_i \in \mathscr{C}$ such that $|X_i| = 1$ or $|Y_i| = 1$. W.l.o.g. assume that $|X_i| = 1$, then neither overlapping vertices $s_{i-1}$ nor $s_i$ can be in $X_i$. Formally, for all $C_i \in \mathscr{C}$,
	$$(|X_i| = 1 \implies s_{i-1} \notin X_i \wedge s_{i} \notin X_i) \wedge (|Y_i| = 1 \implies s_{i-1} \notin Y_i \wedge s_{i} \notin Y_i).$$
\end{lemma}
\begin{proof}
	Assume that a vertex set $C_i \in \mathscr{C}$ exists such that 
	$$\big(|X_i| = 1 \wedge (s_{i-1} \in X_i \vee s_{i} \in X_i)\big) \vee \big(|Y_i| = 1 \wedge (s_{i-1} \in Y_i \vee s_{i} \in Y_i)\big).$$ 
	W.l.o.g. assume that $|X_i| = 1 \wedge s_{i} \in X_i$. Then $G[C_i \cup Y_{i+1}]$ is a complete bigraph. This contradicts the maximality of $C_i$. Thus, there exists no $C_i \in \mathscr{C}$ such that 
	\begin{align}
	\big(|X_i| = 1 \wedge (s_{i-1} \in X_i \vee s_{i} \in X_i)\big) \vee \big(|Y_i| = 1 \wedge (s_{i-1} \in Y_i \vee s_{i} \in Y_i)\big). & 
	\end{align}
\end{proof}

\begin{lemma}\label{ap-rpigimb1}
	Given a chained complete bigraph $G = (X,Y,E)$ with corresponding MCB-component family $\mathscr{C} = \{C_1, \dots, C_n\}$, we have that
	\begin{align*}
	I(G) &\leq \sum\limits_{i=1}^{n}|X_i|\cdot|Y_i| + (|X_i| \mod 2) \cdot (|Y_i|\mod 2) \\
	&- \Bigg(\sum\limits_{i=1}^{n-1}g(s_i,C_i) + g(s_i,C_{i+1})\Bigg) + \Bigg(\sum\limits_{i=1}^{n-1}|g(s_i,C_i) - g(s_i,C_{i+1})|\Bigg).
	\end{align*}
\end{lemma}
We prove that the minimum imbalance of $G$ is bounded from above by the above expression, by constructing an ordering $\sigma_{X \cup Y}$ whose imbalance is equivalent to the above expression.

The ordering $\sigma_{X \cup Y}$ is constructed by creating a subordering for each $C_i \in \mathscr{C}$ separately and concatenating those suborderings. 

\subsubsection{Construction of ordering $\mathbb{\sigma}_{X \cup Y}$}\label{ap-constrordrpibg}
We shall first describe the construction of the subordering corresponding to the vertex set $C_i \in \mathscr{C}$, where $2 \leq i \leq n - 1$. Afterwards, we explain the additional steps required to apply the same construction for vertex sets $C_1$ and $C_n$. For $C_i \in \mathscr{C}$, with $2 \leq i \leq n - 1$, we have the following cases based on the signs of the cardinalities of the parts of $G[C_i]$:
\begin{case}\label{ap-constructordercase1}
	Both $|X_i|$ and $|Y_i|$ are even.\\
	Consider the following cases based on the set membership of the overlapping vertices $s_{i-1}$ and $s_{i}$.
	\begin{enumerate}
		\item Overlapping vertices $s_{i-1}$ and $s_{i}$ are contained in the same part:\\
		W.l.o.g. assume that $s_{i-1}, s_{i} \in X_i$. Let us partition $X_i \setminus \{s_{i-1}, s_i\}$ into $X_1$ and $X_2$ such that $|X_1| = |X_2|$. We define $\sigma_{C_i\setminus \{s_{i-1}, s_i\}}$ as follows:
		$$\sigma_{C_i\setminus \{s_{i-1}, s_i\}} = \sigma_{X_1}\sigma_{Y_i}\sigma_{X_2},$$
		where $\sigma_{X_1}$, $\sigma_{Y_i}$, and $\sigma_{X_2}$ are arbitrary orderings on $X_1$, $Y_i$, and $X_2$ respectively.
		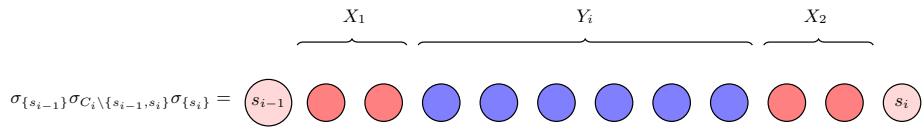
\begin{figure}[H]
			\centering
			\resizebox{\textwidth}{!}{%
\begin{tikzpicture}[-,semithick]

\tikzset{XS/.append style={fill=red!15,draw=black,text=black,shape=circle,minimum size=2em,inner sep=2pt}}
\tikzset{YS/.append style={fill=blue!15,draw=black,text=black,shape=circle,minimum size=2em,inner sep=2pt}}
\tikzset{Y/.append style={fill=blue!50,draw=black,text=black,shape=circle,minimum size=2em,inner sep=2pt}}
\tikzset{X/.append style={fill=red!50,draw=black,text=black,shape=circle,minimum size=2em,inner sep=2pt}}
\tikzset{t/.append style={fill=white,draw=white,text=black}}
\node[t]         (T) {$\sigma_{\{s_{i-1}\}}\sigma_{C_i\setminus \{s_{i-1}, s_i\}}\sigma_{\{s_i\}} =~~~~~~~~~~~~~~~~~~~~~~~~~~~~$};
\node[XS]         (M) [right of=T] {$s_{i-1}$};
\node[X]         (N) [right of=M] {};
\node[X]         (A) [right of=N] {};
\node[Y]         (B) [right of=A] {};
\node[Y]         (C) [right of=B] {};
\node[Y]         (D) [right of=C] {};
\node[Y]         (E) [right of=D] {};
\node[Y]         (F) [right of=E] {};
\node[Y]         (G) [right of=F] {};
\node[X]         (H) [right of=G] {};
\node[X]         (I) [right of=H] {};
\node[XS]         (O) [right of=I] {$s_i$};

\draw[decorate,decoration={brace,amplitude=3pt}] 
(1.5,1) coordinate (t_k_unten) -- (3.4,1) coordinate (t_k_opt_unten); 
\node[t] at (2.5,1.5) {$X_1$};
\draw[decorate,decoration={brace,amplitude=3pt}] 
(3.6,1) coordinate (t_k_unten) -- (9.4,1) coordinate (t_k_opt_unten); 
\node[t] at (6.5,1.5) {$Y_i$};
\draw[decorate,decoration={brace,amplitude=3pt}] 
(9.6,1) coordinate (t_k_unten) -- (11.5,1) coordinate (t_k_opt_unten); 
\node[t] at (10.5,1.5) {$X_2$};
\end{tikzpicture}}%
			\caption{Example of an ordering where $s_{i-1}, s_{i} \in X_i$ and $G[C_i] = K_{6,6}$.}
		\end{figure}
		
		\item Overlapping vertices $s_{i-1}$ and $s_{i}$ are contained in the different parts:\\
		W.l.o.g. assume that $s_{i-1}\in X_i \wedge s_{i} \in Y_i$. Let us partition $X_i \setminus \{s_{i-1}\}$ into $X_1$ and $X_2$ such that $|X_1|+1 = |X_2|$. We also partition $Y_i \setminus \{s_i\}$ into $Y_1$ and $Y_2$ such that $|Y_1| = |Y_2|+1$. We define $\sigma_{C_i\setminus \{s_{i-1}, s_i\}}$ as follows:
		$$\sigma_{C_i\setminus \{s_{i-1}, s_i\}} = \sigma_{X_1}\sigma_{Y_1}\sigma_{X_2}\sigma_{Y_2},$$
		where $\sigma_{X_1}$, $\sigma_{Y_1}$, $\sigma_{X_2}$, and $\sigma_{Y_2}$ are arbitrary orderings on $X_1$, $Y_1$, $X_2$, and $Y_2$ respectively.
		\begin{figure}[H]
			\centering
			\resizebox{\textwidth}{!}{%
\begin{tikzpicture}[-,semithick]

\tikzset{XS/.append style={fill=red!15,draw=black,text=black,shape=circle,minimum size=2em,inner sep=2pt}}
\tikzset{YS/.append style={fill=blue!15,draw=black,text=black,shape=circle,minimum size=2em,inner sep=2pt}}
\tikzset{Y/.append style={fill=blue!50,draw=black,text=black,shape=circle,minimum size=2em,inner sep=2pt}}
\tikzset{X/.append style={fill=red!50,draw=black,text=black,shape=circle,minimum size=2em,inner sep=2pt}}
\tikzset{t/.append style={fill=white,draw=white,text=black}}
\node[t]         (T) {$\sigma_{\{s_{i-1}\}}\sigma_{C_i\setminus \{s_{i-1}, s_i\}}\sigma_{\{s_i\}} =~~~~~~~~~~~~~~~~~~~~~~~~~~~~$};
\node[XS]         (M) [right of=T] {$s_{i-1}$};
\node[X]         (N) [right of=M] {};
\node[X]         (A) [right of=N] {};
\node[Y]         (B) [right of=A] {};
\node[Y]         (C) [right of=B] {};
\node[Y]         (D) [right of=C] {};
\node[X]         (E) [right of=D] {};
\node[X]         (F) [right of=E] {};
\node[X]         (G) [right of=F] {};
\node[Y]         (H) [right of=G] {};
\node[Y]         (I) [right of=H] {};
\node[YS]         (O) [right of=I] {$s_i$};

\draw[decorate,decoration={brace,amplitude=3pt}] 
(1.5,1) coordinate (t_k_unten) -- (3.4,1) coordinate (t_k_opt_unten); 
\node[t] at (2.5,1.5) {$X_1$};
\draw[decorate,decoration={brace,amplitude=3pt}] 
(3.6,1) coordinate (t_k_unten) -- (6.4,1) coordinate (t_k_opt_unten); 
\node[t] at (5,1.5) {$Y_1$};
\draw[decorate,decoration={brace,amplitude=3pt}] 
(6.6,1) coordinate (t_k_unten) -- (9.4,1) coordinate (t_k_opt_unten); 
\node[t] at (8,1.5) {$X_2$};
\draw[decorate,decoration={brace,amplitude=3pt}] 
(9.6,1) coordinate (t_k_unten) -- (11.5,1) coordinate (t_k_opt_unten); 
\node[t] at (10.5,1.5) {$Y_2$};
\end{tikzpicture}}%
			\caption{Example of an ordering where $s_{i-1}\in X_i \wedge s_{i} \in Y_i$ and $G[C_i] = K_{6,6}$.}
		\end{figure}
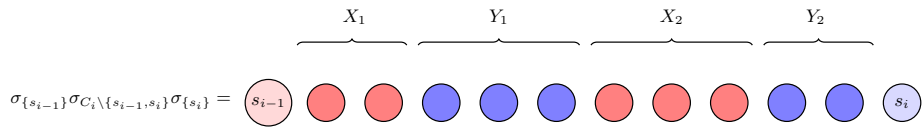
	\end{enumerate}
\end{case}

\begin{case}\label{ap-constructordercase2}
	Either $|X_i|$ is odd or $|Y_i|$ is odd.\\
	W.l.o.g. assume that $|X_i|$ is odd and $|Y_i|$ is even.
	Consider the following cases based on the set membership of the overlapping vertices $s_{i-1}$ and $s_{i}$.
	\begin{enumerate}
		\item Overlapping vertices $s_{i-1}$ and $s_{i}$ are contained in $X_i$:\\
		Let $x \in X_i \setminus \{s_{i-1}, s_i\}$ be chosen arbitrarily. By the fact that $|X_i|$ is odd, such a vertex $x$ must exist.
		Let us partition $X_i \setminus \{s_{i-1}, s_i, x\}$ into $X_1$ and $X_2$ such that $|X_1| = |X_2|$. We also partition $Y_i$ into $Y_1$ and $Y_2$ such that $|Y_1| = |Y_2|$. We define $\sigma_{C_i\setminus \{s_{i-1}, s_i\}}$ as follows:
		$$\sigma_{C_i\setminus \{s_{i-1}, s_i\}} = \sigma_{X_1}\sigma_{Y_1}\sigma_{\{x\}}\sigma_{Y_2}\sigma_{X_2},$$
		where $\sigma_{X_1}$, $\sigma_{Y_1}$, $\sigma_{\{x\}}$, $\sigma_{Y_2}$, and $\sigma_{X_2}$ are arbitrary orderings on $X_1$, $Y_1$, $\{x\}$, $Y_2$, and $X_2$ respectively.
		\begin{figure}[H]
			\centering
			\resizebox{\textwidth}{!}{%
\begin{tikzpicture}[-,semithick]

\tikzset{XS/.append style={fill=red!15,draw=black,text=black,shape=circle,minimum size=2em,inner sep=2pt}}
\tikzset{YS/.append style={fill=blue!15,draw=black,text=black,shape=circle,minimum size=2em,inner sep=2pt}}
\tikzset{Y/.append style={fill=blue!50,draw=black,text=black,shape=circle,minimum size=2em,inner sep=2pt}}
\tikzset{X/.append style={fill=red!50,draw=black,text=black,shape=circle,minimum size=2em,inner sep=2pt}}
\tikzset{t/.append style={fill=white,draw=white,text=black}}
\node[t]         (T) {$\sigma_{\{s_{i-1}\}}\sigma_{C_i\setminus \{s_{i-1}, s_i\}}\sigma_{\{s_i\}} =~~~~~~~~~~~~~~~~~~~~~~~~~~~~$};
\node[XS]        (M) [right of=T] {$s_{i-1}$};
\node[X]         (N) [right of=M] {};
\node[X]         (A) [right of=N] {};
\node[Y]         (B) [right of=A] {};
\node[Y]         (C) [right of=B] {};
\node[Y]         (D) [right of=C] {};
\node[X]         (E) [right of=D] {};
\node[Y]         (F) [right of=E] {};
\node[Y]         (G) [right of=F] {};
\node[Y]         (H) [right of=G] {};
\node[X]         (I) [right of=H] {};
\node[X]         (O) [right of=I] {};
\node[XS]        (P) [right of=O] {$s_i$};

\draw[decorate,decoration={brace,amplitude=3pt}] 
(1.5,1) coordinate (t_k_unten) -- (3.4,1) coordinate (t_k_opt_unten); 
\node[t] at (2.5,1.5) {$X_1$};
\draw[decorate,decoration={brace,amplitude=3pt}] 
(3.6,1) coordinate (t_k_unten) -- (6.4,1) coordinate (t_k_opt_unten); 
\node[t] at (5,1.5) {$Y_1$};
\draw[decorate,decoration={brace,amplitude=3pt}] 
(6.6,1) coordinate (t_k_unten) -- (7.4,1) coordinate (t_k_opt_unten); 
\node[t] at (7,1.5) {$\{x\}$};
\draw[decorate,decoration={brace,amplitude=3pt}] 
(7.6,1) coordinate (t_k_unten) -- (10.4,1) coordinate (t_k_opt_unten); 
\node[t] at (9,1.5) {$Y_2$};
\draw[decorate,decoration={brace,amplitude=3pt}] 
(10.6,1) coordinate (t_k_unten) -- (12.5,1) coordinate (t_k_opt_unten); 
\node[t] at (11.5,1.5) {$X_2$};
\end{tikzpicture}}%
			\caption{Example of an ordering  where $s_{i-1}, s_{i} \in X_i$ and $G[C_i] = K_{7,6}$.}
		\end{figure}
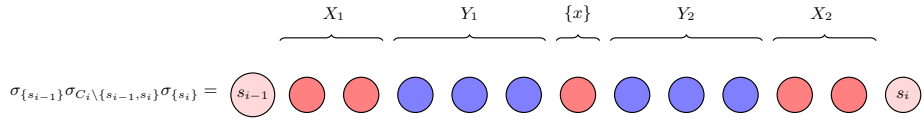
		
		\item Overlapping vertices $s_{i-1}$ and $s_{i}$ are contained in $Y_i$:\\
		Let us partition $Y_i \setminus \{s_{i-1}, s_i\}$ into $Y_1$ and $Y_2$ such that $|Y_1| = |Y_2|$. We define $\sigma_{C_i\setminus \{s_{i-1}, s_i\}}$ as follows:
		$$\sigma_{C_i\setminus \{s_{i-1}, s_i\}} = \sigma_{Y_1}\sigma_{X_i}\sigma_{Y_2},$$
		where $\sigma_{Y_1}$, $\sigma_{X_i}$, and $\sigma_{Y_2}$ are arbitrary orderings on $Y_1$, $X_i$, and $Y_2$ respectively.
		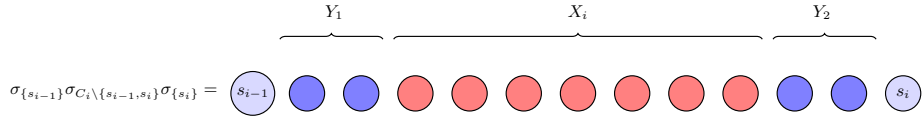
\begin{figure}[H]
			\centering
			\resizebox{\textwidth}{!}{%
\begin{tikzpicture}[-,semithick]

\tikzset{XS/.append style={fill=red!15,draw=black,text=black,shape=circle,minimum size=2em,inner sep=2pt}}
\tikzset{YS/.append style={fill=blue!15,draw=black,text=black,shape=circle,minimum size=2em,inner sep=2pt}}
\tikzset{Y/.append style={fill=blue!50,draw=black,text=black,shape=circle,minimum size=2em,inner sep=2pt}}
\tikzset{X/.append style={fill=red!50,draw=black,text=black,shape=circle,minimum size=2em,inner sep=2pt}}
\tikzset{t/.append style={fill=white,draw=white,text=black}}
\node[t]         (T) {$\sigma_{\{s_{i-1}\}}\sigma_{C_i\setminus \{s_{i-1}, s_i\}}\sigma_{\{s_i\}} =~~~~~~~~~~~~~~~~~~~~~~~~~~~~$};
\node[YS]        (M) [right of=T] {$s_{i-1}$};
\node[Y]         (N) [right of=M] {};
\node[Y]         (A) [right of=N] {};
\node[X]         (B) [right of=A] {};
\node[X]         (C) [right of=B] {};
\node[X]         (D) [right of=C] {};
\node[X]         (E) [right of=D] {};
\node[X]         (F) [right of=E] {};
\node[X]         (G) [right of=F] {};
\node[X]         (H) [right of=G] {};
\node[Y]         (I) [right of=H] {};
\node[Y]         (O) [right of=I] {};
\node[YS]        (P) [right of=O] {$s_i$};

\draw[decorate,decoration={brace,amplitude=3pt}] 
(1.5,1) coordinate (t_k_unten) -- (3.4,1) coordinate (t_k_opt_unten); 
\node[t] at (2.5,1.5) {$Y_1$};
\draw[decorate,decoration={brace,amplitude=3pt}] 
(3.6,1) coordinate (t_k_unten) -- (10.4,1) coordinate (t_k_opt_unten); 
\node[t] at (7,1.5) {$X_i$};
\draw[decorate,decoration={brace,amplitude=3pt}] 
(10.6,1) coordinate (t_k_unten) -- (12.5,1) coordinate (t_k_opt_unten); 
\node[t] at (11.5,1.5) {$Y_2$};
\end{tikzpicture}}%
			\caption{Example of an ordering  where $s_{i-1}, s_{i} \in Y_i$ and $G[C_i] = K_{7,6}$.}
		\end{figure}
		
		\item Overlapping vertices $s_{i-1}$ and $s_{i}$ are contained in different parts:\\
		W.l.o.g. assume that $s_{i-1} \in X_i$ and $s_{i} \in Y_i$.
		Let us partition $X_i \setminus \{s_{i-1}\}$ into $X_1$ and $X_2$ such that $|X_1|+2 = |X_2|$. We also partition $Y_i \setminus \{s_i\}$ into $Y_1$ and $Y_2$ such that $|Y_1| = |Y_2|+1$. We define $\sigma_{C_i\setminus \{s_{i-1}, s_i\}}$ as follows:
		$$\sigma_{C_i\setminus \{s_{i-1}, s_i\}} = \sigma_{X_1}\sigma_{Y_1}\sigma_{X_2}\sigma_{Y_2},$$
		where $\sigma_{X_1}$, $\sigma_{Y_1}$, $\sigma_{X_2}$, and $\sigma_{Y_2}$ are arbitrary orderings on $X_1$, $Y_1$, $X_2$, and $Y_2$ respectively.
		\begin{figure}[H]
			\centering
			\resizebox{\textwidth}{!}{%
\begin{tikzpicture}[-,semithick]

\tikzset{XS/.append style={fill=red!15,draw=black,text=black,shape=circle,minimum size=2em,inner sep=2pt}}
\tikzset{YS/.append style={fill=blue!15,draw=black,text=black,shape=circle,minimum size=2em,inner sep=2pt}}
\tikzset{Y/.append style={fill=blue!50,draw=black,text=black,shape=circle,minimum size=2em,inner sep=2pt}}
\tikzset{X/.append style={fill=red!50,draw=black,text=black,shape=circle,minimum size=2em,inner sep=2pt}}
\tikzset{t/.append style={fill=white,draw=white,text=black}}
\node[t]         (T) {$\sigma_{\{s_{i-1}\}}\sigma_{C_i\setminus \{s_{i-1}, s_i\}}\sigma_{\{s_i\}} =~~~~~~~~~~~~~~~~~~~~~~~~~~~~$};
\node[XS]        (M) [right of=T] {$s_{i-1}$};
\node[X]         (N) [right of=M] {};
\node[X]         (A) [right of=N] {};
\node[Y]         (B) [right of=A] {};
\node[Y]         (C) [right of=B] {};
\node[Y]         (D) [right of=C] {};
\node[X]         (E) [right of=D] {};
\node[X]         (F) [right of=E] {};
\node[X]         (G) [right of=F] {};
\node[X]         (H) [right of=G] {};
\node[Y]         (I) [right of=H] {};
\node[Y]         (O) [right of=I] {};
\node[YS]        (P) [right of=O] {$s_i$};

\draw[decorate,decoration={brace,amplitude=3pt}] 
(1.5,1) coordinate (t_k_unten) -- (3.4,1) coordinate (t_k_opt_unten); 
\node[t] at (2.5,1.5) {$X_1$};
\draw[decorate,decoration={brace,amplitude=3pt}] 
(3.6,1) coordinate (t_k_unten) -- (6.4,1) coordinate (t_k_opt_unten); 
\node[t] at (5,1.5) {$Y_1$};
\draw[decorate,decoration={brace,amplitude=3pt}] 
(6.6,1) coordinate (t_k_unten) -- (10.4,1) coordinate (t_k_opt_unten); 
\node[t] at (8.5,1.5) {$X_2$};
\draw[decorate,decoration={brace,amplitude=3pt}] 
(10.6,1) coordinate (t_k_unten) -- (12.5,1) coordinate (t_k_opt_unten); 
\node[t] at (11.5,1.5) {$Y_2$};
\end{tikzpicture}}%
			\caption{Example of an ordering  where $s_{i-1} \in X_i \wedge s_{i} \in Y_i$ and $G[C_i] = K_{7,6}$.}
		\end{figure}
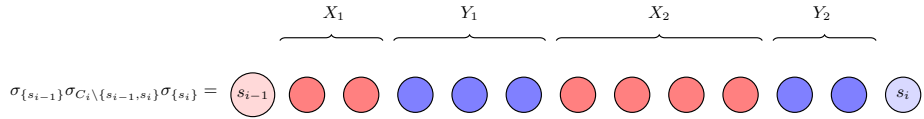
	\end{enumerate}
\end{case}
\begin{case}\label{ap-constructordercase3}
	Both $|X_i|$ and $|Y_i|$ are odd.\\
	Consider the following cases based on the set membership of the overlapping vertices $s_{i-1}$ and $s_{i}$.
	\begin{enumerate}
		\item Overlapping vertices $s_{i-1}$ and $s_{i}$ are contained in the same part:\\
		W.l.o.g. assume that $s_{i-1}, s_{i} \in X_i$. Let us choose $x_m \in X_i \setminus \{s_{i-1}, s_i\}$ arbitrarily. Let $X_1$ and $X_2$ partition $X_i \setminus \{s_{i-1}, s_i, x_m\}$ such that $|X_1| = |X_2|$. Additionally, let $Y_1$ and $Y_2$ partition $Y$ such that $||Y_1|-|Y_2||=1$. We define $\sigma_{C_i\setminus \{s_{i-1}, s_i\}}$ as follows:
		$$\sigma_{C_i\setminus \{s_{i-1}, s_i\}} = \sigma_{X_1}\sigma_{Y_1}\sigma_{\{x_m\}}\sigma_{Y_2}\sigma_{X_2},$$
		where $\sigma_{X_1}$, $\sigma_{Y_1}$, $\sigma_{Y_2}$,  and $\sigma_{X_2}$ are arbitrary orderings on $X_1$, $Y_1$, $Y_2$, and $X_2$ respectively.
		\begin{figure}[H]
			\centering
			\resizebox{\textwidth}{!}{%
	\begin{tikzpicture}[-,semithick]
	
	\tikzset{XS/.append style={fill=red!15,draw=black,text=black,shape=circle,minimum size=2em,inner sep=2pt}}
	\tikzset{YS/.append style={fill=blue!15,draw=black,text=black,shape=circle,minimum size=2em,inner sep=2pt}}
	\tikzset{Y/.append style={fill=blue!50,draw=black,text=black,shape=circle,minimum size=2em,inner sep=2pt}}
	\tikzset{X/.append style={fill=red!50,draw=black,text=black,shape=circle,minimum size=2em,inner sep=2pt}}
	\tikzset{t/.append style={fill=white,draw=white,text=black}}
	\node[t]         (T) {$\sigma_{\{s_{i-1}\}}\sigma_{C_i\setminus \{s_{i-1}, s_i\}}\sigma_{\{s_i\}} =~~~~~~~~~~~~~~~~~~~~~~~~~~~~$};
	\node[XS]        (M) [right of=T] {$s_{i-1}$};
	\node[X]         (N) [right of=M] {};
	\node[X]         (A) [right of=N] {};
	\node[Y]         (B) [right of=A] {};
	\node[Y]         (C) [right of=B] {};
	\node[Y]         (D) [right of=C] {};
	\node[X]         (E) [right of=D] {};
	\node[Y]         (F) [right of=E] {};
	\node[Y]         (G) [right of=F] {};
	\node[Y]         (H) [right of=G] {};
	\node[Y]         (I) [right of=H] {};
	\node[X]         (O) [right of=I] {};
	\node[X]        (P) [right of=O] {};
	\node[XS]        (Q) [right of=P] {$s_i$};
	
	\draw[decorate,decoration={brace,amplitude=3pt}] 
	(1.5,1) coordinate (t_k_unten) -- (3.4,1) coordinate (t_k_opt_unten); 
	\node[t] at (2.5,1.5) {$X_1$};
	\draw[decorate,decoration={brace,amplitude=3pt}] 
	(3.6,1) coordinate (t_k_unten) -- (6.4,1) coordinate (t_k_opt_unten); 
	\node[t] at (5,1.5) {$Y_1$};
	\draw[decorate,decoration={brace,amplitude=3pt}] 
	(6.6,1) coordinate (t_k_unten) -- (7.4,1) coordinate (t_k_opt_unten); 
	\node[t] at (7,1.5) {$\{x\}$};
	\draw[decorate,decoration={brace,amplitude=3pt}] 
	(7.6,1) coordinate (t_k_unten) -- (11.4,1) coordinate (t_k_opt_unten); 
	\node[t] at (9.5,1.5) {$Y_2$};
	\draw[decorate,decoration={brace,amplitude=3pt}] 
	(11.6,1) coordinate (t_k_unten) -- (13.5,1) coordinate (t_k_opt_unten); 
	\node[t] at (12.5,1.5) {$X_2$};
	\end{tikzpicture}}%
			\caption{Example of an ordering  where $s_{i-1}, s_{i} \in X_i$ and $G[C_i] = K_{7,7}$.}
		\end{figure}
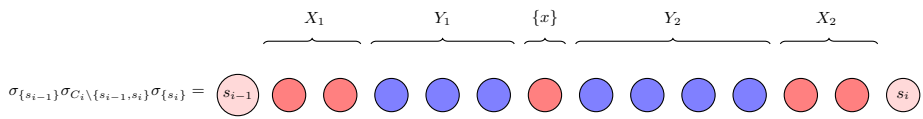
		
		\item Overlapping vertices $s_{i-1}$ and $s_{i}$ are contained in the different parts:\\
		W.l.o.g. assume that $s_{i-1}\in X_i \wedge s_{i} \in Y_i$. Let us partition $X_i \setminus \{s_{i-1}\}$ into $X_1$ and $X_2$ such that $|X_1| = |X_2|$. We also partition $Y_i \setminus \{s_i\}$ into $Y_1$ and $Y_2$ such that $|Y_1| = |Y_2|$. We define $\sigma_{C_i\setminus \{s_{i-1}, s_i\}}$ as follows:
		
		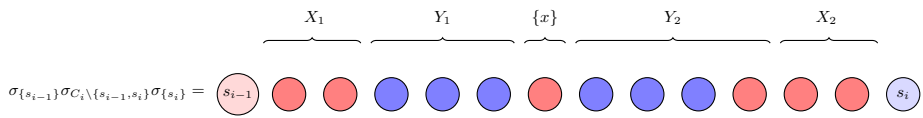
\begin{figure}[H]
			\centering
			\resizebox{\textwidth}{!}{%
	\begin{tikzpicture}[-,semithick]
	
	\tikzset{XS/.append style={fill=red!15,draw=black,text=black,shape=circle,minimum size=2em,inner sep=2pt}}
	\tikzset{YS/.append style={fill=blue!15,draw=black,text=black,shape=circle,minimum size=2em,inner sep=2pt}}
	\tikzset{Y/.append style={fill=blue!50,draw=black,text=black,shape=circle,minimum size=2em,inner sep=2pt}}
	\tikzset{X/.append style={fill=red!50,draw=black,text=black,shape=circle,minimum size=2em,inner sep=2pt}}
	\tikzset{t/.append style={fill=white,draw=white,text=black}}
	\node[t]         (T) {$\sigma_{\{s_{i-1}\}}\sigma_{C_i\setminus \{s_{i-1}, s_i\}}\sigma_{\{s_i\}} =~~~~~~~~~~~~~~~~~~~~~~~~~~~~$};
	\node[XS]        (M) [right of=T] {$s_{i-1}$};
	\node[X]         (N) [right of=M] {};
	\node[X]         (A) [right of=N] {};
	\node[Y]         (B) [right of=A] {};
	\node[Y]         (C) [right of=B] {};
	\node[Y]         (D) [right of=C] {};
	\node[X]         (E) [right of=D] {};
	\node[Y]         (F) [right of=E] {};
	\node[Y]         (G) [right of=F] {};
	\node[Y]         (H) [right of=G] {};
	\node[X]         (I) [right of=H] {};
	\node[X]         (O) [right of=I] {};
	\node[X]        (P) [right of=O] {};
	\node[YS]        (Q) [right of=P] {$s_i$};
	
	\draw[decorate,decoration={brace,amplitude=3pt}] 
	(1.5,1) coordinate (t_k_unten) -- (3.4,1) coordinate (t_k_opt_unten); 
	\node[t] at (2.5,1.5) {$X_1$};
	\draw[decorate,decoration={brace,amplitude=3pt}] 
	(3.6,1) coordinate (t_k_unten) -- (6.4,1) coordinate (t_k_opt_unten); 
	\node[t] at (5,1.5) {$Y_1$};
	\draw[decorate,decoration={brace,amplitude=3pt}] 
	(6.6,1) coordinate (t_k_unten) -- (7.4,1) coordinate (t_k_opt_unten); 
	\node[t] at (7,1.5) {$\{x\}$};
	\draw[decorate,decoration={brace,amplitude=3pt}] 
	(7.6,1) coordinate (t_k_unten) -- (11.4,1) coordinate (t_k_opt_unten); 
	\node[t] at (9.5,1.5) {$Y_2$};
	\draw[decorate,decoration={brace,amplitude=3pt}] 
	(11.6,1) coordinate (t_k_unten) -- (13.5,1) coordinate (t_k_opt_unten); 
	\node[t] at (12.5,1.5) {$X_2$};
	\end{tikzpicture}}%
			\caption{Example of an ordering  where $s_{i-1} \in X_i \wedge s_{i} \in Y_i$ and $G[C_i] = K_{7,7}$.}
		\end{figure}
	\end{enumerate}
\end{case}
The vertex set $C_1$ contains exactly one overlapping vertex, namely $s_1$. We apply the above described construction to create $\sigma_{C_1}$ by arbitrarily picking a vertex in $C_1 \setminus \{s_1\}$ to be $s_0$. Similarly, for the vertex set $C_n$, we construct $\sigma_{C_n}$ by arbitrarily picking a vertex in $C_n \setminus \{s_{n-1}\}$ to be $s_n$ and applying the above construction.

Note that each constructed subordering $\sigma_{\{s_{i-1}\}}\sigma_{C_i\setminus \{s_{i-1}, s_i\}}\sigma_{\{s_i\}}$ is imbalance optimal for the graph induced by the vertices of $C_i$. That is, $\forall_{C_i\in \mathscr{C}} I(G[C_i]) = I(G[C_i], \sigma_{\{s_{i-1}\}}\sigma_{C_i\setminus \{s_{i-1}, s_i\}}\sigma_{\{s_i\}})$. 

Using the above constructed suborderings, we define ordering $\sigma_{X \cup Y}$ as:
\\\\
\begin{align*}
	\sigma_{X \cup Y} = &\sigma_{\{s_0\}}\sigma_{C_1\setminus\{s_0,s_1\}}\sigma_{\{s_1\}}\sigma_{C_2\setminus\{s_1,s_2\}}\sigma_{\{s_2\}}\dots \\
	&\sigma_{C_{n-1}\setminus\{s_{n-2},s_{n-1}\}}\sigma_{\{s_{n-1}\}}\sigma_{C_n\setminus\{s_{n-1},s_n\}}\sigma_{\{s_n\}}.
\end{align*}

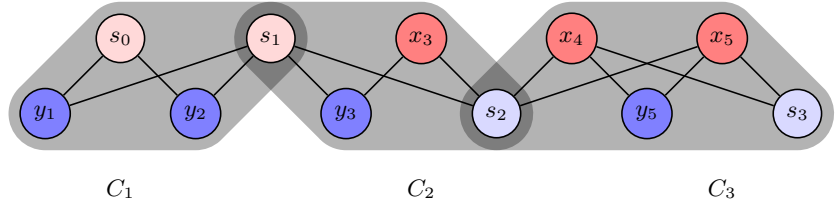
\begin{figure}[H]
	\centering
	\begin{tikzpicture}[-,semithick]
\tikzset{XS/.append style={fill=red!15,draw=black,text=black,shape=circle}}
\node[XS]         (x2) at (1,1) {$s_1$};

\tikzset{YS/.append style={fill=blue!15,draw=black,text=black,shape=circle}}
\node[YS]         (y4) at (4,0) {$s_2$};

\tikzset{Y/.append style={fill=blue!50,draw=black,text=black,shape=circle}}
\node[Y]         (y1) at (-2,0) {$y_1$};
\node[Y]         (y2) at (0,0) {$y_2$};
\node[Y]         (y3) at (2,0) {$y_3$};
\node[Y]         (y5) at (6,0) {$y_5$};
\node[YS]         (y6) at (8,0) {$s_3$};

\tikzset{X/.append style={fill=red!50,draw=black,text=black,shape=circle}}
\node[XS]         (x1) at (-1,1) {$s_0$};
\node[X]         (x3) at (3,1) {$x_3$};
\node[X]         (x4) at (5,1) {$x_4$};
\node[X]         (x5) at (7,1) {$x_5$};

\tikzset{t/.append style={fill=white,draw=white,text=black}}
\node[t]         at (-1,-1) {$C_1$};
\node[t]         at (3,-1) {$C_2$};
\node[t]         at (7,-1) {$C_3$};

\path (x1) edge              node {} (y1)
edge              node {} (y2)
(x2) edge              node {} (y1)
edge              node {} (y2)
edge              node {} (y3)
edge              node {} (y4)
(x3) edge              node {} (y3)
edge              node {} (y4)
(x4) edge              node {} (y5)
edge              node {} (y6)
edge              node {} (y4)
(x5) edge              node {} (y5)
edge              node {} (y6)
edge              node {} (y4);

\begin{scope}[on background layer]     
\fill[black,opacity=0.3] \convexpath{x1,x2,y2,y1}{1.5em}; 
\fill[black,opacity=0.3] \convexpath{x2,x3,y4,y3}{1.5em}; 
\fill[black,opacity=0.3] \convexpath{x4,x5,y6,y5,y4}{1.5em};
\end{scope}
\end{tikzpicture}
	\caption{Example of a chained complete bigraph with assigned $s_0$, and $s_n$.}
\end{figure}
\begin{figure}[H]
	\centering
	\resizebox{0.95\textwidth}{!}{%
\begin{tikzpicture}[-,semithick]
  
  \tikzset{XS/.append style={fill=red!15,draw=black,text=black,shape=circle}}
  \tikzset{YS/.append style={fill=blue!15,draw=black,text=black,shape=circle}}
  \tikzset{Y/.append style={fill=blue!50,draw=black,text=black,shape=circle}}
  \tikzset{X/.append style={fill=red!50,draw=black,text=black,shape=circle}}
  \node at (-1,0) {$\sigma_{C_1} = $};
  \node[Y]         (x1) at (0,0) {$y_1$};
  \node[Y]         (y2) at (1,0) {$y_2$};
  
  \node at (3,0) {$\sigma_{C_2} = $};
  \node[Y]         (y3) at (4,0) {$y_3$};
  \node[X]         (y4) at (5,0) {$x_3$};
  
  \node at (7,0) {$\sigma_{C_3} = $};
  \node[X]         (y5) at (8,0) {$x_4$};
  \node[Y]         (x4) at (9,0) {$y_5$};
  \node[X]         (y6) at (10,0) {$x_5$};
  
  \node at (-1,-1) {$\sigma_{X \cup Y} = $};

  \node[XS]         (y1) at (0.25,-1) {$s_0$};
  
  \node[Y]         (x1) at (1.25,-1) {$y_1$};
  \node[Y]         (y2) at (2.25,-1) {$y_2$};
  
  \node[XS]         (x1) at (3.25,-1) {$s_1$};
  
  \node[Y]         (y3) at (4.25,-1) {$y_3$};
  \node[X]         (y4) at (5.25,-1) {$x_3$};
  
  \node[YS]         (x1) at (6.25,-1) {$s_2$};
  
  \node[X]         (y5) at (7.25,-1) {$x_4$};
  \node[Y]         (x4) at (8.25,-1) {$y_5$};
  \node[X]         (y6) at (9.25,-1) {$x_5$};
  
  \node[YS]         (x1) at (10.25,-1) {$s_3$};
\end{tikzpicture}

}%
	\caption{Constructed suborderings and final ordering of the example graph above.}
\end{figure}
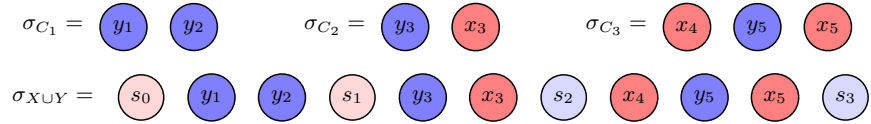

\subsubsection{Proof of \cref{ap-rpigimb1}}\label{ap-analysisconstrordsec}
To prove \cref{ap-rpigimb1} we require several remarks on the constructed ordering $\sigma_{X \cup Y}$.
\setcounter{remark}{4}
\begin{remark}\label{ap-ordanalysisrem3}
	From the definition of chained complete bigraphs, a vertex $v \in (X \cup Y) \setminus S$ is contained in exactly one vertex set $C_i \in \mathscr{C}$. Additionally, by \cref{astrtrip}, the neighborhood of $v$ is a subset of $C_i$. Therefore, the imbalance of $v$ is only influenced by the vertices in $C_i$. That is,
	$\forall_{C_i \in \mathscr{C}} ~\forall_{v \in C_i \setminus S} ~I(v, \sigma_{X\cup Y}, G) = I(v, \sigma_{X\cup Y}^{C_i}, G[C_i]).$
\end{remark}
\begin{remark}\label{ap-ordanalysisrem4}
	From the definition of chained complete bigraphs, an overlapping vertex $s_i \in S$ is contained in exactly two vertex sets $C_i, C_{i+1} \in \mathscr{C}$. Additionally, by \cref{astrtrip}, the neighborhood of $s_i$ is a subset of $C_i \cup C_{i+1}$. Therefore, the imbalance of $s_i$ is only influenced by the vertices in $C_i \cup C_{i+1}$. That is,
	$\forall_{s_i \in S} ~I(s_{i}, \sigma_{X\cup Y}, G) = I(s_i, \sigma_{X\cup Y}^{C_{i} \cup C_{i+1}}, G[C_i \cup C_{i+1}]).$
	From the position of overlapping vertex $s_i$ in ordering $\sigma_{\{s_{i-1}\}}\sigma_{C_i\setminus \{s_{i-1}, s_i\}}\sigma_{\{s_{i}\}}\sigma_{C_{i+1}\setminus \{s_{i}, s_{i+1}\}}\sigma_{\{s_{i+1}\}}$, we deduce that the imbalance of $s_i$ in ordering 
	$$\sigma_{\{s_{i-1}\}}\sigma_{C_i\setminus \{s_{i-1}, s_i\}}\sigma_{\{s_{i}\}}\sigma_{C_{i+1}\setminus \{s_{i}, s_{i+1}\}}\sigma_{\{s_{i+1}\}}$$ 
	is $|g(s_i,C_i)- g(s_i,C_{i+1})|$. That is,
	
	\begin{align*}
	\MoveEqLeft I(s_i, \sigma_{\{s_{i-1}\}}\sigma_{C_i\setminus \{s_{i-1}, s_i\}}\sigma_{\{s_{i}\}}\sigma_{C_{i+1}\setminus \{s_{i}, s_{i+1}\}}\sigma_{\{s_{i+1}\}}, G[C_i \cup C_{i+1}]) \\
	& = |g(s_i,C_i)- g(s_i,C_{i+1})|.
	\end{align*}
\end{remark}
\setcounter{remark}{6}
\begin{remark}\label{ap-ordanalysisrem1}
	By applying a similar analysis as in \cref{lmacbg1} together with \cref{rpigimbaux1}, we derive that for all $C_i \in \mathscr{C}$ the ordering $\sigma_{\{s_{i-1}\}}\sigma_{C_i\setminus \{s_{i-1}, s_i\}}\sigma_{\{s_i\}}$ is an imbalance optimal ordering for complete bigraph $G[C_i]$. That is, by \cref{thm1} and the aforementioned analysis, for all $C_i \in \mathscr{C}$ we have
	$$I(\sigma_{\{s_{i-1}\}}\sigma_{C_i\setminus \{s_{i-1}, s_i\}}\sigma_{\{s_i\}}, G[C_i]) = |X_i|\cdot|Y_i| + (|X_i| \mod 2) \cdot (|Y_i|\mod 2).$$
\end{remark}
\begin{remark}\label{ap-ordanalysisrem2}
	Since overlapping vertex $s_i \in S$ is the rightmost vertex in ordering $\sigma_{\{s_{i-1}\}}\sigma_{C_i\setminus \{s_{i-1}, s_i\}}\sigma_{\{s_i\}}$, we deduce that the imbalance of $s_i$ in ordering $\sigma_{\{s_{i-1}\}}\sigma_{C_i\setminus \{s_{i-1}, s_i\}}\sigma_{\{s_i\}}$ is the cardinality of the part of $G[C_i]$ in which $s_i$ is not contained. That is, for all $s_i \in S$ we have
	$$I(s_{i},\sigma_{\{s_{i-1}\}}\sigma_{C_i\setminus \{s_{i-1}, s_i\}}\sigma_{\{s_i\}}, G[C_i]) = g(s_{i},C_i).$$
	
	Similarly, since overlapping vertex $s_i \in S$ is the leftmost vertex in ordering $\sigma_{\{s_{i}\}}\sigma_{C_{i+1}\setminus \{s_{i}, s_{i+1}\}}\sigma_{\{s_{i+1}\}}$, we deduce that the imbalance of $s_i$ in ordering $\sigma_{\{s_{i}\}}\sigma_{C_{i+1}\setminus \{s_{i}, s_{i+1}\}}\sigma_{\{s_{i+1}\}}$ is the cardinality of the part of $G[C_{i+1}]$ in which $s_i$ is not contained. That is, for all $s_i \in S$ we have
	$$I(s_{i},\sigma_{\{s_{i}\}}\sigma_{C_{i+1}\setminus \{s_{i}, s_{i+1}\}}\sigma_{\{s_{i+1}\}}, G[C_{i+1}]) = g(s_{i},C_{i+1}).$$
\end{remark}
\begin{proof}
	The imbalance of the ordering $\sigma_{X \cup Y} =: \sigma_{V}$, constructed using the method of \cref{ap-constrordrpibg}, can be written as follows:
	\begin{align}
	I(\sigma_{V})& = \Bigg(\sum\limits_{v \in (X\cup Y) \setminus S} I(v, \sigma_{X \cup Y}, G)\Bigg) + \Bigg(\sum\limits_{v \in S} I(v, \sigma_{X \cup Y}, G)\Bigg) \notag\\ 
	&= \sum\limits_{i=1}^{n} I(\sigma_{\{s_{i-1}\}}\sigma_{C_i\setminus \{s_{i-1}, s_i\}}\sigma_{\{s_{i}\}}, G[C_i]) \notag\\
	&- \sum\limits_{i=1}^{n-1} I(s_i, \sigma_{\{s_{i-1}\}}\sigma_{C_i\setminus \{s_{i-1}, s_i\}}\sigma_{\{s_i\}}, G[C_i]) \notag \\
	&- \sum\limits_{i=1}^{n-1} I(s_i, \sigma_{\{s_{i}\}}\sigma_{C_{i+1}\setminus \{s_{i}, s_{i+1}\}}\sigma_{\{s_{i+1}\}}, G[C_{i+1}]) \notag\\
	&+ \sum\limits_{s_i \in S} I(s_i, \sigma_{\{s_{i-1}\}}\sigma_{C_i\setminus \{s_{i-1}, s_i\}}\sigma_{\{s_{i}\}}\sigma_{C_{i+1}\setminus \{s_{i}, s_{i+1}\}}\sigma_{\{s_{i+1}\}}, G[C_i \cup C_{i+1}]) \label{ap-constrpibganalysis1}\\
	&= \sum\limits_{i=1}^{n}|X_i|\cdot|Y_i| + (|X_i| \mod 2) \cdot (|Y_i|\mod 2) \notag\\
	&- \Bigg(\sum\limits_{i=1}^{n-1}g(s_i,C_i) + g(s_i,C_{i+1})\Bigg) + \Bigg(\sum\limits_{i=1}^{n-1}|g(s_i,C_i) - g(s_i,C_{i+1})|\Bigg), \label{ap-constrpibganalysis2}
	\end{align}
	where \cref{ap-constrpibganalysis1} follows from \cref{ap-ordanalysisrem3}, \cref{ap-ordanalysisrem4} and the construction method of \cref{ap-constrordrpibg}. \cref{ap-constrpibganalysis2} follows from \cref{ap-ordanalysisrem1} and \cref{ap-ordanalysisrem2}. 
	
	Note that in the expression at \cref{ap-constrpibganalysis1} the first term twice adds an ``incorrect imbalance'' for each vertex in $S$. This ``incorrect imbalance'' is removed by the second and third term of the expression.
\end{proof}

\subsubsection{Proof of the lower bound}
\setcounter{lemma}{6}
\begin{lemma}\label{ap-rpigimb2}
	Given a chained complete bigraph $G = (X,Y,E)$ with corresponding MCB-component family $\mathscr{C} = \{C_1, \dots, C_n\}$, we have that
	\begin{align*}
	I(G) &\geq \sum\limits_{i=1}^{n}|X_i|\cdot|Y_i| + (|X_i| \mod 2) \cdot (|Y_i|\mod 2) \\
	&- \Bigg(\sum\limits_{i=1}^{n-1}g(s_i,C_i) + g(s_i,C_{i+1})\Bigg) + \Bigg(\sum\limits_{i=1}^{n-1}|g(s_i,C_i) - g(s_i,C_{i+1})|\Bigg).
	\end{align*}
\end{lemma}
\begin{proof}
	We shall prove that the imbalance of any arbitrary ordering $\sigma_{X \cup Y}$ on the vertex set $X \cup Y$ is bounded from below by the above expression by induction on $|\mathscr{C}| = n$.
	\begin{itemize}[label=$\bullet$]
		\item Base Case $(n = 0)$:\\
		By the definition of MCB-components $\mathscr{C}$, the graph $G$ is an empty graph. Thus,
		\begin{align*}
		I(G) = 0 &= \sum\limits_{i=1}^{n}|X_i|\cdot|Y_i| + (|X_i| \mod 2) \cdot (|Y_i|\mod 2) \\
		&- \Bigg(\sum\limits_{i=1}^{n-1}g(s_i,C_i) + g(s_i,C_{i+1})\Bigg)+ \Bigg(\sum\limits_{i=1}^{n-1}|g(s_i,C_i) - g(s_i,C_{i+1})|\Bigg).
		\end{align*}
		
		\item Base case $(n = 1)$:\\
		By the definition of MCB-components $\mathscr{C}$, the graph $G$ is a complete bigraph. Thus, by \cref{thm1}
		\begin{align*}
		I(G) & =|X| \cdot |Y| + (|X| \mod 2) \cdot (|Y|\mod 2) \\
		&=\sum\limits_{i=1}^{n}|X_i|\cdot|Y_i| + (|X_i| \mod 2) \cdot (|Y_i|\mod 2) \\
		&-\Bigg(\sum\limits_{i=1}^{n-1}g(s_i,C_i) + g(s_i,C_{i+1})\Bigg)+ \Bigg(\sum\limits_{i=1}^{n-1}|g(s_i,C_i) - g(s_i,C_{i+1})|\Bigg).
		\end{align*}
		
		\item Induction hypothesis:\\
		For any chained complete bigraph $G = (X,Y,E)$ with corresponding MCB-component family $\mathscr{C} = \{C_1, \dots, C_k\}$, where $0 \leq k < n$, it holds that
		\begin{align*}
		I(G) & \geq \sum\limits_{i=1}^{k}|X_i|\cdot|Y_i| + (|X_i| \mod 2) \cdot (|Y_i|\mod 2) \\
		&- \Bigg(\sum\limits_{i=1}^{k-1}g(s_i,C_i) + g(s_i,C_{i+1})\Bigg)+ \Bigg(\sum\limits_{i=1}^{k-1}|g(s_i,C_i) - g(s_i,C_{i+1})|\Bigg).
		\end{align*}
		
		\item Induction step $(n>1)$:\\
		Let $G = (X,Y,E)$ be a chained complete bigraph with corresponding MCB-component family $\mathscr{C} = \{C_1, \dots, C_{k+1}\}$. Let us define $\mathscr{C} \setminus C_{k+1}$ as follows:
		$$\mathscr{C} \setminus C_{k+1} = \bigcup\limits_{C_i \in \mathscr{C}\setminus \{C_{k+1}\}}C_i.$$ 
		
		\begin{figure}[H]
			\centering
			\begin{tikzpicture}[-,semithick]
\tikzset{XS/.append style={fill=red!15,draw=black,text=black,shape=circle}}
\node[XS]         (x2) at (1,1) {$s_1$};

\tikzset{YS/.append style={fill=blue!15,draw=black,text=black,shape=circle}}
\node[YS]         (y4) at (4,0) {$s_2$};

\tikzset{Y/.append style={fill=blue!50,draw=black,text=black,shape=circle}}
\node[Y]         (y1) at (-2,0) {$y_1$};
\node[Y]         (y2) at (0,0) {$y_2$};
\node[Y]         (y3) at (2,0) {$y_3$};
\node[Y]         (y5) at (6,0) {$y_5$};
\node[Y]         (y6) at (8,0) {$y_6$};

\tikzset{X/.append style={fill=red!50,draw=black,text=black,shape=circle}}
\node[X]         (x1) at (-1,1) {$x_1$};
\node[X]         (x3) at (3,1) {$x_3$};
\node[X]         (x4) at (5,1) {$x_4$};
\node[X]         (x5) at (7,1) {$x_5$};

\tikzset{t/.append style={fill=white,draw=white,text=black}}
\node[t]         at (1,-1) {$\mathscr{C} \setminus C_{k+1}$};
\node[t]         at (6,-1) {$C_{k+1}$};

\path (x1) edge              node {} (y1)
edge              node {} (y2)
(x2) edge              node {} (y1)
edge              node {} (y2)
edge              node {} (y3)
edge              node {} (y4)
(x3) edge              node {} (y3)
edge              node {} (y4)
(x4) edge              node {} (y5)
edge              node {} (y6)
edge              node {} (y4)
(x5) edge              node {} (y5)
edge              node {} (y6)
edge              node {} (y4);

\begin{scope}[on background layer]     
\fill[black,opacity=0.3] \convexpath{x1,x2,x3,y4,y3,y2,y1}{1.5em}; 
\fill[black,opacity=0.3] \convexpath{x4,x5,y6,y5,y4}{1.5em};
\end{scope}
\end{tikzpicture}
			\caption{Illustration of the definitions of $\mathscr{C} \setminus C_{k+1}$ and $C_{k+1}$.}
		\end{figure}
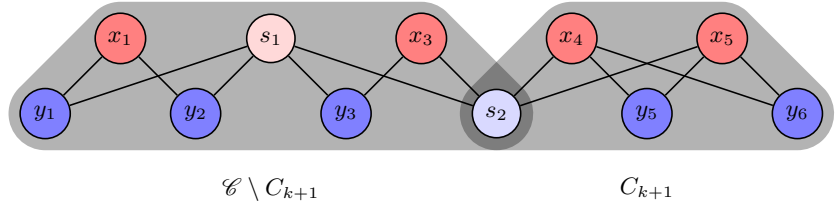
		Using the above definition, we write the imbalance of $\sigma_{X \cup Y}$ as follows:
		\begin{align}
		I(\sigma_{X \cup Y}) & = I(\sigma_{X\cup Y}^{\mathscr{C} \setminus C_{k+1}}, G[\mathscr{C} \setminus C_{k+1}]) + I(\sigma_{X\cup Y}^{C_{k+1}}, G[C_{k+1}]) \notag\\
		&- I(s_{k}, \sigma_{X \cup Y}^{C_{k}}, G[C_{k}]) - I(s_{k}, \sigma_{X \cup Y}^{C_{k+1}}, G[C_{k+1}]) \notag\\ 
		& + I(s_{k}, \sigma_{X \cup Y}^{C_{k} \cup C_{k+1}}, G[C_{k} \cup C_{k+1}]) \label{ap-eqlem7-0}\\
		&\geq I(G[\mathscr{C} \setminus C_{k+1}]) + I( G[C_{k+1}]) - I(s_{k}, \sigma_{X \cup Y}^{C_{k}}, G[C_{k}]) \notag\\
		& - I(s_{k}, \sigma_{X \cup Y}^{C_{k+1}}, G[C_{k+1}]) + I(s_{k}, \sigma_{X \cup Y}^{C_{k} \cup C_{k+1}}, G[C_{k} \cup C_{k+1}]) \notag\\
		&\geq\Bigg(\sum\limits_{i=1}^{k}|X_i|\cdot|Y_i| + (|X_i| \mod 2) \cdot (|Y_i|\mod 2)\Bigg)\notag\\
		&- \Bigg(\sum\limits_{i=1}^{k-1}g(s_i,C_i) + g(s_i,C_{i+1})\Bigg) + \Bigg(\sum\limits_{i=1}^{k-1}|g(s_i,C_i) - g(s_i,C_{i+1})|\Bigg) \notag\\
		&+ |X_{k+1}|\cdot|Y_{k+1}| + (|X_{k+1}| \mod 2) \cdot (|Y_{k+1}|\mod 2) \notag\\
		&- I(s_{k}, \sigma_{X \cup Y}^{C_{k}}, G[C_{k}])  - I(s_{k}, \sigma_{X \cup Y}^{C_{k+1}}, G[C_{k+1}]) \notag\\
		&+ I(s_{k}, \sigma_{X \cup Y}^{C_{k} \cup C_{k+1}}, G[C_{k} \cup C_{k+1}]) \label{ap-eqlem7-1}\\
		&=\Bigg(\sum\limits_{i=1}^{k+1}|X_i|\cdot|Y_i| + (|X_i| \mod 2) \cdot (|Y_i|\mod 2)\Bigg) \notag\\
		&- \Bigg(\sum\limits_{i=1}^{k-1}g(s_i,C_i) + g(s_i,C_{i+1})\Bigg) + \Bigg(\sum\limits_{i=1}^{k-1}|g(s_i,C_i) - g(s_i,C_{i+1})|\Bigg) \notag\\
		&- I(s_{k}, \sigma_{X \cup Y}^{C_{k}}, G[C_{k}]) - I(s_{k}, \sigma_{X \cup Y}^{C_{k+1}}, G[C_{k+1}]) \notag\\
		&+ I(s_{k}, \sigma_{X \cup Y}^{C_{k} \cup C_{k+1}}, G[C_{k} \cup C_{k+1}])\notag\\
		& \geq\Bigg(\sum\limits_{i=1}^{k+1}|X_i|\cdot|Y_i| + (|X_i| \mod 2) \cdot (|Y_i|\mod 2)\Bigg) \notag\\
		& - \Bigg(\sum\limits_{i=1}^{k-1}g(s_i,C_i) + g(s_i,C_{i+1})\Bigg) + \Bigg(\sum\limits_{i=1}^{k-1}|g(s_i,C_i) - g(s_i,C_{i+1})|\Bigg)\notag\\
		&- g(s_{k}, C_{k+1}) - g(s_{k}, C_{k}) + |g(s_{k}, C_{k}) - g(s_{k}, C_{k+1})| \label{ap-eqlem7-2}\\
		&= \Bigg(\sum\limits_{i=1}^{k+1}|X_i|\cdot|Y_i| + (|X_i| \mod 2) \cdot (|Y_i|\mod 2)\Bigg)\notag\\
		&- \Bigg(\sum\limits_{i=1}^{k}g(s_i,C_i) + g(s_i,C_{i+1})\Bigg) + \Bigg(\sum\limits_{i=1}^{k}|g(s_i,C_i) - g(s_i,C_{i+1})|\Bigg), \notag
		\end{align}
		where \cref{ap-eqlem7-0} follows from \cref{ap-ordanalysisrem3} and \cref{ap-ordanalysisrem4}, \cref{ap-eqlem7-1} follows from the induction hypothesis, and \cref{ap-eqlem7-2} follows from \cref{rpigimbaux}. 
		
		Similar to the analysis in \cref{ap-analysisconstrordsec}, in the expression at \cref{ap-eqlem7-0}, the first term twice adds an ``incorrect imbalance'' of vertex $s_{k}$. The ``incorrect imbalance'' is removed by the second and third term of the expression. The fourth term of the expression adds the correct imbalance of $s_{k}$.
	\end{itemize}
\end{proof}

\end{document}